\documentclass[journal]{IEEEtran}
\usepackage{amsmath,amssymb,amsfonts,bm}
\usepackage{algorithm, algorithmic}
\usepackage{breqn}

\usepackage{array}
\usepackage[caption=false,font=normalsize,labelfont=sf,textfont=sf]{subfig}
\usepackage{textcomp}
\usepackage{stfloats}
\usepackage{url}
\usepackage{verbatim}
\usepackage{graphicx}
\usepackage{float}
\usepackage{color}
\usepackage{subfig}
\usepackage{cite}
\usepackage{threeparttable}
\newtheorem{theorem}{Theorem}
\newtheorem{lemma}{Lemma}
\newtheorem{proposition}{Proposition}
\newtheorem{remark}{Remark}
\newenvironment{proof}{{\indent \indent \it Proof:}}{\hfill $\blacksquare$\par}

\newcolumntype{P}[1]{>{\centering\arraybackslash}p{#1}} 
\newcommand{\tr}{\operatorname{Tr}}

\begin{document}

\title{Rethinking WMMSE: Can Its Complexity Scale Linearly With the Number of BS Antennas?}

\author{Xiaotong Zhao,
 Siyuan Lu, 
Qingjiang Shi, and Zhi-Quan Luo
        

\thanks{Xiaotong Zhao is with the School of Software Engineering,
Tongji University, Shanghai 201804, China (e-mail: xiaotongzhao@tongji.edu.cn).}
\thanks{Siyuan Lu was with the School of Software Engineering,
Tongji University, Shanghai 201804, China. He is now with Huawei Technology (e-mail: lusiyuan666@163.com).}
\thanks{Qingjiang Shi is with the School of Software Engineering, Tongji University,
Shanghai 201804, China, and also with Shenzhen Research Institute of Big Data, Shenzhen 518172, China (e-mail: shiqj@tongji.edu.cn).}

\thanks{Zhi-Quan Luo is with the Chinese University of Hong Kong, Shenzhen 518172, China, and also with Shenzhen Research Institute of Big Data, Shenzhen 518172, China (e-mail: luozq@cuhk.edu.cn).}}


\IEEEpubid{0000--0000/00\$00.00~\copyright~2021 IEEE}

\maketitle

\begin{abstract}
Precoding design for maximizing weighted sum-rate (WSR) is a fundamental problem for downlink of massive multi-user multiple-input multiple-output (MU-MIMO) systems. It is well-known that this problem is generally NP-hard due to the presence of multi-user interference. The weighted minimum mean-square error (WMMSE) algorithm is a popular approach for WSR maximization. However, its computational complexity is cubic in the number of base station (BS) antennas, which is unaffordable when the BS is equipped with a large antenna array. In this paper, we consider the WSR maximization problem with either a sum-power constraint (SPC) or per-antenna power constraints (PAPCs). For the former, we prove that any nontrivial stationary point must have a low-dimensional subspace structure, and then propose a reduced-WMMSE (R-WMMSE) with linear complexity by exploiting the solution structure. For the latter, we propose a linear-complexity WMMSE approach, named PAPC-WMMSE, by using a novel recursive design of the algorithm. Both R-WMMSE and PAPC-WMMSE have simple closed-form updates and guaranteed convergence to stationary points. Simulation results verify the efficacy of the proposed designs, especially the much lower complexity as compared to the state-of-the-art approaches for massive MU-MIMO systems.
\end{abstract}

\section{Introduction}
\IEEEPARstart{M}{assive} MU-MIMO is one of the key enabling technologies for the fifth-generation and next-generation networks\cite{zhang2020prospective,marzetta2016fundamentals,wang2019overview,de2022overview}. For MU-MIMO systems, a fundamental problem is to design transmit precoders that maximize the system weighted sum-rate (WSR) subject to power constraints. However, it is well-known that the WSR maximization problem is NP-hard\cite{luo2008dynamic,liu2011coordinated}. Meanwhile, since the base station (BS) would have hundreds or thousands of antennas in massive MU-MIMO systems\cite{zhang2020prospective,marzetta2016fundamentals,wang2019overview,de2022overview}, the computational complexity of precoding remains a big challenge. Therefore, it is highly desirable to have an efficient precoding algorithm with a complexity that scales linearly with the number of BS antennas at most while still having guaranteed convergence to stationary points of the WSR maximization problem. This is the focus of this paper.

There are two classes power constraints for the WSR maximization problems, i.e., \emph{sum power constraint} (SPC) and \emph{per-antenna power constraints} (PAPCs).

\subsubsection{\textbf{WSR Maximization With SPC}} Most of the existing works on WSR maximization considered SPC, i.e., the total transmit power of all BS antennas is not allowed to exceed a given power budget. Among these works, a few are dedicated to globally solving the WSR problems, e.g., \cite{joshi2012weighted} for MU-MISO systems, \cite{liu2012achieving} for Gaussian interference channels, etc. However, these global methods have exponential computational complexity, which is prohibitive for practical MU-MIMO systems. In practice, sub-optimal precoding methods with low complexity are preferred. Some of them are well-recognized in practical MU-MIMO systems, such as maximum ratio transmission (MRT) \cite{kammoun2014linear}, zero-forcing (ZF) \cite{gao2011linear}, regularized ZF precoding (RZF) \cite{nguyen2019multi}, etc. Although these suboptimal methods have low complexity, they are not aimed to directly solve the WSR maximization problem and generally come with significant performance loss. In contrast to suboptimal methods, iterative precoding algorithms try to directly solve the WSR maximization problem and thus can achieve a good balance between the WSR maximization and the computational complexity.

\IEEEpubidadjcol

There are mainly two classes of iterative algorithms for the WSR maximization problem. One is the successive convex approximation (SCA) method. The authors in \cite{shi2009monotonic} and \cite{kim2011optimal} sequentially constructed convex surrogates of the non-convex WSR objective and solved the resulting convex problems to increase the WSR. The SCA-based algorithm ensures convergence to a stationary point of the WSR maximization problem. 
Several variants of the SCA-based method were also proposed for various scenarios in \cite{ng2010linear,tran2012fast,nguyen2014sum}. The other class of iterative precoding methods is the classical weighted minimum mean-square error (WMMSE) algorithm \cite{christensen2008TWC,shi2011iteratively}. The idea behind the WMMSE is the relationship between mean-square error (MSE) and signal to interference plus noise ratio (SINR). By using the MSE-SINR relationship, the WSR maximization problem can be equivalently transformed into a weighted sum-MSE minimization problem, which is then iteratively solved by applying the block coordinate descent (BCD) method, leading to the WMMSE algorithm with three closed-form updates. The WMMSE algorithm is guaranteed to reach a stationary point of the original WSR maximization problem. Moreover, the WMMSE generally performs much faster than the SCA-based algorithm because of the simple closed-form updates which allow parallel implementation. As a result, the WMMSE algorithm has been widely used in spectral efficiency optimization of various communication systems \cite{sun2018learning,zhao2021exploiting,ghauch2017sum}. However, due to the required matrix inversion operation for precoders update, the computational complexity of WMMSE is \emph{cubic} in the number of BS antennas and thus is unaffordable for practical massive MU-MIMO systems with hundreds or thousands of antennas. To the best of our knowledge, there exists no linear-complexity precoding algorithm that can guarantee convergence to stationary points of the WSR maximization problem under SPC.

\subsubsection{\textbf{WSR Maximization With PAPCs}} Since the power amplifier of each BS antenna has its own power budget \cite{yu2007transmitter}, it is more practical to consider the WSR maximization problem with PAPCs. A straightforward way to tackle this problem is to first derive the precoders under the SPC and then downscale them to satisfy the PAPCs\cite{lee2013zero}. However, such a strategy is rather heuristic and may incur a significant performance degradation. Some other works \cite{li2015parallel,jang2015low,pham2017efficient} were dedicated to developing variants of ZF precoders under PAPCs for MU-MISO or MU-MIMO systems. Since ZF can remove the multi-user interference, the corresponding WSR maximization problem with PAPCs becomes convex and thus can be globally solved. However, these ZF-based variants can not achieve the maximum spectral efficiency because they are not aimed to directly solve the original WSR maximization problem.

In addition to the above suboptimal methods, researchers have also proposed some iterative algorithms to directly solve the non-convex WSR maximization problem with PAPCs. The authors in \cite{Mao2019rate} investigated the WSR maximization problem for rate-splitting-multiple-access-based downlink coordinated multi-point joint transmission networks subject to both the individual quality of service rate constraints and the PAPCs. They used the same technique as in WMMSE \cite{shi2011iteratively} to transform the original WSR problem into a weighted sum-MSE minimization problem and solved it using BCD. In the BCD iteration for the precoder update under PAPCs, the authors applied the interior-point method to solve the corresponding subproblem, a convex quadratically constrained quadratic program (QCQP). In \cite{shi2008per}, the authors considered jointly optimizing transmit powers, transmit filters, and receive filters for MU-MIMO systems by leveraging the well-known SINR-MSE relationship. Due to the coupling of variables, the problem was decomposed into several subproblems, whereby the power optimization subproblem was globally solved by a geometric program while the transmitter optimization subproblem was solved by a second-order cone program, resulting in high computational complexity. In \cite{thomas2018hybrid}, the MU-MIMO hybrid beamforming (HBF) problem for WSR maximization under PAPCs was investigated by majorization-minimization and alternating optimization, involving high-dimensional matrix operation with cubic complexity in the number of BS antennas.

To summarize, the existing iterative precoding algorithms for WSR maximization with SPC or PAPCs generally require high complexity operations (e.g., high dimensional matrix inversion), resulting in a cubic complexity in the number of BS antennas. This makes the existing algorithms unsuitable for implementation in MIMO systems with a large antenna array. Is there an MU-MIMO precoding method with a linear complexity (in the number of BS antennas) and a theoretical convergence guarantee to stationary points of the WSR maximization problem? This paper answers this question in the affirmative. Specifically, we investigate the WSR maximization problem with either SPC or PAPCs, and propose linear-complexity WMMSE approaches to solve them with guaranteed convergence to stationary points. Here, `linear complexity' means that the computational complexity of the proposed approaches is linear in the number of BS antennas.

The main contributions of this paper are two-fold.

\begin{itemize}
\item[1)] {\textbf{R-WMMSE With Linear Complexity for the SPC Case:}}
For an MU-MIMO system with $M\gg N \geq D\geq K$, where $M$ denotes the number of BS antennas, $N$ denotes the total number of user antennas, and $D$ denotes the total number of data streams sent from the BS to $K$ users, we prove the first key property, termed  \emph{low-dimensional subspace property}, that any nontrivial stationary point of the WSR maximization problem under SPC must lie in the range space of $\mathbf{H}^H$, where $\mathbf{H}\in\mathbb{C}^{N\times M}$ denotes the channel matrix between the BS and users. Further, we prove the second key property that any nontrivial stationary point must satisfy the SPC with equality, which is termed {\em full power property}. Using the low-dimensional subspace property, we reduce the WSR problem to a much lower dimensional decision space, thus successfully avoiding high dimensional matrix operations of the WMMSE algorithm. Furthermore, by exploiting the full power property, the reduced WSR maximization problem can be further transformed into an unconstrained problem, making the follow-up BCD iteration easier to implement. The obtained variant of WMMSE is named reduced-WMMSE (R-WMMSE). The complexity of R-WMMSE is in the order of $\mathcal{O}(M)$, and the per-iteration complexity is even independent of $M$. Moreover, by using block matrix operations, it is shown that the R-WMMSE requires only $D$-dimensional matrix inversions during each iteration. Similar to the WMMSE algorithm, the R-WMMSE is provably convergent to at least stationary points of the original WSR maximization problem. Numerical experiments show that the R-WMMSE achieves the same WSR performance as the WMMSE, but performs significantly faster than the WMMSE when the BS is equipped with a large antenna array. For example, when $M=1024$, the R-WMMSE is $100+$ times faster than the WMMSE.


\item[2)] {\textbf{PAPC-WMMSE With Linear Complexity for the PAPCs Case:}}
For the WSR maximization problem with PAPCs, we first use the WMMSE framework to transform the problem into the weighted sum-MSE minimization problem and then apply BCD to solve it, yielding three subproblems including transmitter optimization (i.e., precoder optimization), receiver optimization and weight matrix optimization. Due to the PAPCs, the subproblem of precoder optimization becomes a convex quadratic program with multiple quadratic constraints, which appears difficult at first glance. By applying BCD with the precoder variables further split into $M$ blocks (each corresponding to per antenna) and exploiting the particular problem structure, we obtain a closed-form update for each block variable. Furthermore, by a novel recursive design of the algorithm, we eventually obtain a linear-complexity algorithm, termed PAPC-WMMSE, with guaranteed convergence to stationary points. Finally, numerical experiments show that the proposed PAPC-WMMSE algorithm performs much better than the state-of-the-art algorithms in terms of both WSR performance and CPU time.


\end{itemize} 

The rest of the paper is organized as follows. Section II presents the downlink system model and problem formulation. Section III briefly reviews the classical WMMSE framework. Section IV proposes R-WMMSE under SPC and provides the corresponding convergence result.
Section V proposes PAPC-WMMSE under PAPCs and analyzes its convergence result. Comprehensive numerical experiments are provided in Section VI. Finally, Section VII concludes the paper.

\emph{Notation:} Throughout this paper, scalars are denoted by both lower and upper case letters, while vectors and matrices are denoted by boldface lower case and boldface upper case letters, respectively. $\Re e(a)$ is the real part of a complex scalar $a$. The space of $M\times N$ complex matrices is expressed as $\mathbb{C}^{M\times N}$. For a matrix $\mathbf{A}$, $\mathbf{A}^{T}$, $\mathbf{A}^{\ast}$, $\mathbf{A}^{H}$, $\mathbf{A}^{-1}$, $\text{Tr}(\mathbf{A})$, and $R(\mathbf{A})$ signify its transpose, conjugate, conjugate transpose, inverse, trace, and range space, respectively. The Euclidean norm of a vector $\mathbf{a}$ is defined as $\|\mathbf{a}\|_{2}=\sqrt{\mathbf{a}^{H}\mathbf{a}}$. The Frobenius norm of a matrix $\mathbf{A}$ is defined by $\|\mathbf{A}\|_{F}=\sqrt{\text{Tr}\left(\mathbf{A}^{H}\mathbf{A}\right)}$. $\mathbf{I}$ denotes the identity matrix, and $\text{blkdiag}(\mathbf{A}_1,\ldots,\mathbf{A}_K)$ denotes a block diagonal matrix with $\mathbf{A}_1,\ldots,\mathbf{A}_K$ as the diagonal blocks. Finally, the distribution of a circularly symmetric complex Gaussian random vector with mean $\mu$ and covariance matrix $\mathbf{\Sigma}$ is represented by $\mathcal{CN}(\mu,\mathbf{\Sigma})$.

\section{System Model and Problem Formulation}

\subsection{Downlink System Model}
Consider a downlink massive MU-MIMO system where a BS with $M$ transmit antennas simultaneously serves $K$ users each with $N_k$ receive antennas. Let $\mathbf{s}_k \in \mathbb{C}^{D_k\times 1}$ denote the symbol vector intended for user $k$ and $\mathbf{P}_k \in \mathbb{C}^{M \times D_k}$ denote the linear precoder for user $k$. Then the transmitted signal of BS can be expressed as
\begin{equation}\label{BS_signal}
    \mathbf{x}\triangleq\sum_{k=1}^{K}\mathbf{P}_k\mathbf{s}_k,
\end{equation}
where it is assumed that $\mathbf{s}_k \sim \mathcal{CN}\left(\mathbf{0},\mathbf{I}\right)$.

Under the flat-fading channel model assumption, the received signal at user $k$ can be expressed by
\begin{equation}\label{system_model_yk}
\begin{aligned}
\mathbf{y}_{k} &=\mathbf{H}_{k} \mathbf{x}+\mathbf{n}_{k} \\
&=\underbrace{\mathbf{H}_{k} \mathbf{P}_{k} \mathbf{s}_{k}}_{\text {desired signal of user } k}+\underbrace{\sum_{j=1, j \neq k}^{K} \mathbf{H}_{k} \mathbf{P}_{j} \mathbf{s}_{j}}_{\text {multi-user interference }}+\mathbf{n}_{k}, \ \forall k,
\end{aligned}
\end{equation}
where $\mathbf{H}_{k} \in \mathbb{C}^{N_k \times M}$ denotes the channel matrix from the BS to user $k$, $\mathbf{n}_{k} \in \mathbb{C}^{N_k \times 1}$ is the additive white Gaussian noise vector with distribution $\mathcal{CN}\left(\mathbf{0},\sigma_k^2\mathbf{I}\right)$. Moreover, it is assumed that the transmitted symbol vectors for different users are independent of each other as well as the noise vectors.

Let us define $N\triangleq \sum_{k=1}^{K}N_k$ and $D\triangleq\sum_{k=1}^{K}D_k\geq K$. Furthermore, define $\mathbf{y}\triangleq[\mathbf{y}_{1}^{T},\mathbf{y}_{2}^{T},\ldots,\mathbf{y}_{K}^{T}]^{T}\in \mathbb{C}^{N\times 1}$, $\mathbf{H}\triangleq[\mathbf{H}_{1}^{T}, \mathbf{H}_{2}^{T},$ $\ldots,\mathbf{H}_{K}^{T}]^{T}\in \mathbb{C}^{N\times M}$, $\mathbf{P}\triangleq[\mathbf{P}_{1},\mathbf{P}_{2},\ldots,\mathbf{P}_{K}]\in \mathbb{C}^{M\times D}$, $\mathbf{s}\triangleq[\mathbf{s}_{1}^{T},\mathbf{s}_{2}^{T},\ldots,\mathbf{s}_{K}^{T}]^{T}\in \mathbb{C}^{D\times 1}$, and $\mathbf{n}\triangleq[\mathbf{n}_{1}^{T},$ $\mathbf{n}_{2}^{T},\ldots,\mathbf{n}_{K}^{T}]^{T}\in \mathbb{C}^{N\times 1}$. Then Eq. \eqref{system_model_yk} can be written in a more compact form as follows
\begin{equation}\label{system_model_y_compact}
\mathbf{y}=\mathbf{H}\mathbf{P} \mathbf{s}+\mathbf{n}. 
\end{equation}
We make a very mild assumption that $\mathbf{H}$ has full row rank throughout the rest of this paper. 

\begin{remark}\!\!\emph{:}
For a massive MU-MIMO system, the total number of BS antennas is much larger than the total number of user antennas, i.e., we have $M\gg N\geq D\geq K$, where the last two inequality follows directly. Hence, when the BS is equipped with an extremely large antenna array, an efficient precoding algorithm is expected to have computational complexity linear in $M$ or even independent of $M$.
\end{remark}

\subsection{Problem Formulation}
A fundamental problem is to find the set of optimal precoders $\{\mathbf{P}_k\}_{k=1}^K$ that maximizes the system WSR subject to power constraints. The WSR is given by
\begin{equation}\label{WSR_objective_function}
\begin{aligned}
R=\sum_{k=1}^{K} \alpha_{k} R_k, 
\end{aligned}
\end{equation}
where the weight $\alpha_{k}$ denotes the priority of user $k$, and $R_k$ is the achievable rate of user $k$ given by
\begin{equation}\label{single_user_rate}
\begin{aligned}
R_k\triangleq\log &\operatorname{det}\Bigg(\mathbf{I}+\mathbf{H}_{k} \mathbf{P}_{k} \mathbf{P}_{k}^{H} \mathbf{H}_{k}^{H} \\
&\times \Big(\sum_{j \neq k} \mathbf{H}_{k} \mathbf{P}_{j} \mathbf{P}_{j}^{H} \mathbf{H}_{k}^{H}+\sigma_{k}^{2} \mathbf{I}\Big)^{-1}\Bigg). 
\end{aligned}
\end{equation}

There are two kinds of power constraints, i.e., SPC and PAPCs. They give rise to the following two WSR maximization problems.

\emph{1) WSR Maximization With SPC:}
Most works focus on SPC that limits the total power consumed by all antennas at the BS. Under SPC, the WSR problem can be formulated as follows.
\begin{subequations}\label{wsr_total_problem}
\begin{align}
\max_{\left\{\mathbf{P}_{k}\right\}}\  \ &\sum_{k=1}^{K} \alpha_{k} R_k \label{wsr_total_objective}\\ 
\text { s.t. } \  \ & \sum_{k=1}^{K} \operatorname{Tr}\left(\mathbf{P}_{k} \mathbf{P}_{k}^{H}\right) \leq P_{\max},\label{wsr_total_constraints}
\end{align}
\end{subequations}
where $ P_{\max}$ represents the total transmit power budget of BS.

\emph{2) WSR Maximization With PAPCs:}
Although SPC is widely considered in the literature, considering that the power amplifier of each antenna has its own limit on transmit power, we also consider the following WSR maximization with PAPCs.
\begin{subequations}\label{wsr_PAPC_problem}
\begin{align}
\max_{\left\{\mathbf{P}_{k}\right\}}\  \ & \sum_{k=1}^{K} \alpha_{k} R_{k} \label{WSR_PAPC}\\
\text { s.t. }\  \ & \sum_{k=1}^{K} \left[\mathbf{P}_{k} \mathbf{P}_{k}^{H}\right]_{m,m} \leq P_{m},\ \forall m, \label{eq_PAPC}
\end{align}
\end{subequations}
where $\left[\mathbf{A}\right]_{m,m}$ denotes the $m$-th diagonal element of matrix $\mathbf{A}$. Constraints \eqref{eq_PAPC} indicate that the transmit power at the $m$-th transmit antenna of the BS is not allowed to exceed $P_m$.

Clearly, problem \eqref{wsr_PAPC_problem} is more difficult than problem \eqref{wsr_total_problem} due to the presence of more quadratic constraints. However, the main difficulty of these two problems arises from the highly nonlinear and nonconvex WSR objective function. Moreover, following \cite{luo2008dynamic}, it can be shown that both problems are NP-hard, stated in the following proposition.
\begin{proposition}\label{lemma_NP_hard}
\emph{(WSR Maximization is NP-hard):} Both problem \eqref{wsr_total_problem} and problem \eqref{wsr_PAPC_problem} are NP-hard. 
\end{proposition}


\section{ The Classical WMMSE: A Revisit}\label{section_WMMSE}
The WMMSE framework is widely used for WSR maximization problems. In this section, we introduce the classical WMMSE framework \cite{shi2011iteratively,shi2015secure} from a new perspective, i.e., without the need for a physically meaningful definition of MSE, as stated below.

The key idea behind the WMMSE algorithm is to transform the nonconvex WSR maximization problem into another equivalent tractable weighted MSE minimization problem by introducing auxiliary variables, which can be solved by the BCD method \cite{shi2011iteratively}. Originally, the equivalence between WSR maximization and weighted MSE minimization is established by the MSE-SINR relationship with explicit physical meaning. To derive the equivalence, we need to first define MSE in terms of the channel model and then write down the weighted MSE minimization problem. However, this is unnecessary for the equivalence establishment. Actually, we have the following essential facts summarized in Lemma \ref{lemma_WMMSE}.

\begin{lemma}\label{lemma_WMMSE}
\emph{(Principle Behind the WMMSE Framework \cite{shi2015secure}):} Given matrices $\mathbf{A}\in\mathbb{C}^{n\times p}$, $\mathbf{B}\in\mathbb{C}^{p\times l}$ and any positive definite matrix $\mathbf{N}\in\mathbb{C}^{n\times n}$, we have
\begin{equation}\label{eq_WMMSE_nature}
\begin{aligned}
\log&\det( \mathbf{I}+\mathbf{A} \mathbf{B}\mathbf{B}^H\mathbf{A}^H\mathbf{N}^{-1}) \\
&=\max_{\mathbf{\Omega}\succ \mathbf{0},\mathbf{\Gamma}} \log\det\left(\mathbf{\Omega}\right) -\tr\left(\mathbf{\Omega}\mathbf{E}\left(\mathbf{\Gamma},\mathbf{B}\right)\right)+l,
\end{aligned}
\end{equation}
where $\mathbf{\Gamma}\in\mathbb{C}^{n\times l}$ and positive definite matrix $\mathbf{\Omega}\in\mathbb{C}^{l\times l}$ are two auxiliary variables, and
\begin{equation}
\mathbf{E}\left(\mathbf{\Gamma},\mathbf{B}\right) \triangleq   \left(\mathbf{I}-\mathbf{\Gamma}^H\mathbf{A}\mathbf{B}\right)\left(\mathbf{I}-\mathbf{\Gamma}^H\mathbf{A}\mathbf{B}\right)^H + \mathbf{\Gamma}^H\mathbf{N}\mathbf{\Gamma}\label{eq:MSE}
\end{equation}
is an $l\times l$ matrix function. Meanwhile, the optimal $\mathbf{\Gamma}$ and $\mathbf{\Omega}$ for the right-hand side of \eqref{eq_WMMSE_nature} are respectively given by
\begin{equation}\label{eq_lemma_optimal_U}
\hat{\mathbf{\Gamma}}=\left(\mathbf{N}+\mathbf{A}\mathbf{B}\mathbf{B}^H\mathbf{A}^H\right)^{-1}\mathbf{A}\mathbf{B}
\end{equation}
and
\begin{equation}\label{eq_lemma_optimal_W}
\hat{\mathbf{\Omega}}=\left(\mathbf{E}\left(\hat{\mathbf{\Gamma}},\mathbf{B}\right)\right)^{-1}=\left(\mathbf{I}-\hat{\mathbf{\Gamma}}^H\mathbf{A}\mathbf{B}\right)^{-1}.
\end{equation}
\end{lemma}

Lemma \ref{lemma_WMMSE} can be used to quickly derive the equivalent problem of the WSR maximization problem and the corresponding WMMSE algorithm. Specifically, first, by comparing $R_k$ in \eqref{single_user_rate}  with the left-hand side of \eqref{eq_WMMSE_nature}, we define $\mathbf{A}_k\triangleq\mathbf{H}_k$, $\mathbf{B}_k\triangleq\mathbf{P}_k$, and $\mathbf{N}_k\triangleq\sum_{j \neq k} \mathbf{H}_{k} \mathbf{P}_{j} \mathbf{P}_{j}^{H} \mathbf{H}_{k}^{H}+\sigma_{k}^{2} \mathbf{I}$. Then by applying Lemma \ref{lemma_WMMSE} to the WSR maximization problems, we obtain an equivalent problem as follows
\begin{equation}\label{eq_transform_problem_WMMSE}
\begin{aligned}
\min_{\mathbf{W},\mathbf{U},\mathbf{P}}\  \ & \sum_{k=1}^{K} \alpha_{k} \left(\operatorname{Tr}\left(\mathbf{W}_{k}\mathbf{E}_{k}\right) - \log\det\left(\mathbf{W}_{k}\right) \right) \\
\text { s.t.~~}\  \ & \mathbf{P} \in \mathcal{P}, 
\end{aligned}
\end{equation}
where $\mathbf{P} \in \mathcal{P}$ denotes either SPC or PAPCs, $\mathbf{W}\triangleq\left\{\mathbf{W}_{k}\right\}_{k=1}^{K}$ and $\mathbf{U}\triangleq\left\{\mathbf{U}_{k}\right\}_{k=1}^{K}$ are  auxiliary variables, which play the same roles of $\mathbf{\Omega}$ and $\mathbf{\Gamma}$ in Lemma \ref{lemma_WMMSE}, and
\begin{equation}\label{MSE_matrix_WMMSE}
\begin{aligned}
\mathbf{E}_{k}  \triangleq &(\mathbf{I}-\mathbf{U}_{k}^H \mathbf{H}_{k} \mathbf{P}_{k})(\mathbf{I}- \mathbf{U}_{k}^H \mathbf{H}_{k}\mathbf{P}_{k})^H \\
  & + \mathbf{U}_{k}^H\left(\sum_{j\neq k}  \mathbf{H}_{k} \mathbf{P}_{j}\mathbf{P}_{j}^H \mathbf{H}_{k}^H 
 + \sigma_k^2\mathbf{I}\right) \mathbf{U}_{k}
\end{aligned}
\end{equation}
is obtained by using \eqref{eq:MSE} rather than by a physically meaningful definition of MSE, despite the same form as MSE.

Although problem \eqref{eq_transform_problem_WMMSE} is jointly non-convex over $\left(\mathbf{U},\mathbf{W},\mathbf{P} \right)$, it is convex with respect to each individual variable $\mathbf{U},\mathbf{W},\mathbf{P}$. Thus, the BCD method can be applied to iteratively minimize the weighted sum-MSE cost function, yielding the WMMSE algorithm. Specifically, by invoking the results of Lemma \ref{lemma_WMMSE}, the update of $\mathbf{U}$ while fixing $\mathbf{W}$ and $\mathbf{P}$ is given by
\begin{equation}\label{U_update_WMMSE}
\mathbf{U}_{k} = \left(\sum_{j=1}^K \mathbf{H}_{k}\mathbf{P}_{j}\mathbf{P}_{j}^{H}\mathbf{H}_{k}^{H}+\sigma_{k}^{2} \mathbf{I} \right)^{-1}\mathbf{H}_{k}\mathbf{P}_{k},\; \forall\; k,
\end{equation}
and the update of $\mathbf{W}$ while fixing the other two block variables is given by
\begin{equation}\label{W_update_WMMSE}
\mathbf{W}_{k} = \left(\mathbf{I}-\mathbf{U}_{k}^H\mathbf{H}_k\mathbf{P}_k \right)^{-1},\; \forall\; k.
\end{equation}
While fixing $\mathbf{U}$ and $\mathbf{W}$, the precoder update is obtained by solving the following problem
\begin{equation}\label{eq_P_update_WMMSE}
    \begin{aligned}
        \min_{\mathbf{P}}  ~   & \sum_{k=1}^{K} \alpha_{k}\operatorname{Tr}\left(\mathbf{W}_{k}\left(\mathbf{I}-\mathbf{U}_{k}^{H}\mathbf{H}_{k} \mathbf{P}_{k}\right)\left(\mathbf{I}-\mathbf{U}_{k}^{H}\mathbf{H}_{k} \mathbf{P}_{k}\right)^H  \right)\\
       &+ \sum_{k=1}^{K}\alpha_{k}\operatorname{Tr}\left(\mathbf{W}_{k}\sum_{j\neq k} \mathbf{U}_{k}^{H}\mathbf{H}_{k} \mathbf{P}_{j}\mathbf{P}_{j}^{H}\mathbf{H}_{k}^{H}\mathbf{U}_{k}\right)     \\
        \text{ s.t.~} \     & \mathbf{P} \in \mathcal{P}.
    \end{aligned}
\end{equation}

By integrating the above updates, a detailed description of the WMMSE algorithm is given in Algorithm \ref{alg_WMMSE}. It is seen that line 3 and line 4 require only small-scale matrix inversion operations because both the number of symbols $D_k$  and the number of receive antennas $N_k$ of user $k$ are small (typically, we have $D_k\leq N_k\leq 4$ for user equipment in mobile communications). However, line 5 requires an $M$-dimensional matrix inversion operation for the SPC case (see \cite{shi2011iteratively} for more details) or calling interior-point methods, resulting in a computational complexity of at least $\mathcal{O}(M^3)$, which is unaffordable when $M$ is extremely large in massive MU-MIMO systems.


\begin{algorithm}[!tb]
    \renewcommand{\algorithmicrequire}{\textbf{Initialization:}}
    \renewcommand{\algorithmicensure}{\textbf{Output:}}
    \caption{The WMMSE Framework}  \label{alg_WMMSE}
    \begin{algorithmic}[1]
    \REQUIRE Initialize $\mathbf{P}$ such that $\mathbf{P} \in \mathcal{P}$ and $\mathbf{W}_k=\mathbf{I},\forall k$. Set the tolerance of accuracy $\epsilon$.
    \REPEAT
    \STATE $\mathbf{W}'_{k} = \mathbf{W}_{k},\; \forall\; k$;\\
    \STATE  $\mathbf{U}_{k} = \left(\sum_{j=1}^K \mathbf{H}_{k}\mathbf{P}_{j}\mathbf{P}_{j}^{H}\mathbf{H}_{k}^{H}+\sigma_{k}^{2} \mathbf{I} \right)^{-1}\mathbf{H}_{k}\mathbf{P}_{k},\; \forall\; k$;
    \STATE   $\mathbf{W}_{k} = \left(\mathbf{I}-\mathbf{U}_{k}^H\mathbf{H}_k\mathbf{P}_k \right)^{-1},\; \forall\; k$;
     \STATE Update $\mathbf{P}$ by solving problem \eqref{eq_P_update_WMMSE};
    \UNTIL $|\sum_k \alpha_k\log\det(\mathbf{W}_k)-\sum_k \alpha_k\log\det(\mathbf{W}'_k)|\leq \epsilon$.
    \ENSURE $\mathbf{P}_k,\forall k$.
    \end{algorithmic}
\end{algorithm}


\section{The Proposed R-WMMSE for the SPC Case}
In this section, we investigate the WSR maximization problem with SPC and propose a variant of WMMSE with linear computational complexity by exploring the structure of stationary points of problem \eqref{wsr_total_problem}. In what follows, we first state two important properties as well as an equivalent problem reformulation. Then we present the R-WMMSE algorithm and study its convergence.
\subsection{Important Properties and Problem Reformulation}
As mentioned in Section \ref{section_WMMSE}, the original WMMSE algorithm for the SPC case in \cite{shi2011iteratively} requires a high-dimensional matrix inversion operation in each iteration, resulting in cubic computational complexity of $\mathcal{O}\left(M^3\right)$. Fortunately, by investigating the structure of stationary points of problem \eqref{wsr_total_problem}, we can finally derive a linear-complexity WMMSE algorithm. Particularly, to distinguish stationary points, we give the following definition. 

\emph{Definition 1} \emph{(Trivial Stationary Point):}
We say a point $\mathbf{P}$ satisfying $\mathbf{H}_k\mathbf{P}_k=\mathbf{0}, \forall k$, which results in a zero WSR, is a trivial stationary point of problem \eqref{wsr_total_problem}.

With the definition of a trivial stationary point, we have the following proposition for nontrivial stationary points of problem \eqref{wsr_total_problem},  which states a critical property.
\begin{proposition}\label{proposition_column_space}
\emph{(Low-Dimensional Subspace Property):} Any nontrivial stationary point $\{\mathbf{P}_{k}^{\star}\}$ of problem \eqref{wsr_total_problem} must lie in the range space of $\mathbf{H}^{H}$, i.e., $\mathbf{P}_{k}^{\star}=\mathbf{H}^{H}\mathbf{X}_{k}$, with some $\mathbf{X}_{k} \in \mathbb{C}^{ N \times  D_{k}}, \forall k $.
\end{proposition}
\begin{proof}
See Appendix \ref{appendix_proposition_column_space}.
\end{proof}

\begin{remark}
\emph{(Classical Precoding Methods Obey the Low-Dimensional Subspace Property):}
In fact, many classical precoding algorithms conform to our optimal precoding structure in Proposition \ref{proposition_column_space}, e.g., the MRT precoding $\mathbf{P}_{\text{MRT}}=\mathbf{H}^{H}$, the ZF precoding $\mathbf{P}_{\text{ZF}}=\mathbf{H}^{H}\left(\mathbf{H}\mathbf{H}^{H}\right)^{-1}$ \cite{parfait2014performance}, the RZF precoding $\mathbf{P}_{\text{RZF}}=\mathbf{H}^{H}\left(\mathbf{H}\mathbf{H}^{H}+\mu \mathbf{I}\right)^{-1}$\cite{peel2005vector}, with $\mu$ is a regularization parameter. In addition, we find that the eigen zero-forcing (EZF) \cite{sun2010eigen} precoding, which is widely used in real-world MU-MIMO systems, also has this structure, as shown in Appendix \ref{appendix_ezf}.
\end{remark}

Proposition \ref{proposition_column_space}  means that the dimension of the decision variable $\mathbf{P}\in \mathbb{C}^{M\times D}$ can be greatly reduced to the size of $\mathbf{X}\triangleq[\mathbf{X}_{1},\mathbf{X}_{2},\ldots,\mathbf{X}_{K}]\in \mathbb{C}^{N\times D}$ given $M\gg N$. That is, by using the low-dimensional subspace property, problem \eqref{wsr_total_problem} can be solved by equivalently solving 
\begin{equation}\label{eq: WSR-X}
\begin{aligned}
\max_{\mathbf{X}}~~ &\sum_{k=1}^K \alpha_k \log \operatorname{det}\Bigg(\mathbf{I}+\mathbf{H}_{k} \mathbf{H}^H\mathbf{X}_{k} \mathbf{X}_{k}^{H}\mathbf{H} \mathbf{H}_{k}^{H} \\
&\times \Big(\sum_{j \neq k} \mathbf{H}_{k} \mathbf{H}^H\mathbf{X}_{j} \mathbf{X}_{j}^{H}\mathbf{H} \mathbf{H}_{k}^{H}+\sigma_{k}^{2} \mathbf{I}\Big)^{-1}\Bigg)\\
\text { ~s.t. } \  \ & \sum_{k=1}^{K} \operatorname{Tr}\left(\mathbf{H}^H\mathbf{X}_{k} \mathbf{X}_{k}^{H}\mathbf{H}\right) \leq P_{\max},
\end{aligned}
\end{equation}
which has a smaller decision space. Although the dimension of the decision space is greatly reduced by using Proposition \ref{proposition_column_space}, the reduced problem \eqref{eq: WSR-X} is still difficult because we need to use the Bisection method to tackle the power constraint when the WMMSE framework is applied to the problem. Fortunately, we find another important property that can be used to eliminate the power constraint. For a clear illustration, we first focus on the WSR maximization problem with SPC and state the full power property in the following proposition.

\begin{proposition}\label{proposition_optimal_with_equality}
\emph{(Full Power Property):}
Any nontrivial stationary point of problem \eqref{wsr_total_problem} must satisfy the sum power constraint \eqref{wsr_total_constraints} with equality.
\end{proposition}
\begin{proof}
This immediately follows from Lemma \ref{lemma_multiplier} in Appendix \ref{appendix_proposition_column_space} and the complementary slackness condition given by \eqref{complementary_slackness}.
\end{proof}

By using Proposition \ref{proposition_optimal_with_equality} and exploiting the fractional structure of SINR, we can reduce problem \eqref{wsr_total_problem} to the  following unconstrained problem.
\begin{equation}\label{unconstrained_WSR_problem}
\begin{aligned}
      & \max_{\{\mathbf{P}_{k}\}}  \  \  \sum_{k=1}^{K} \alpha_{k}\log \operatorname{det}\Bigg(\mathbf{I}+ \mathbf{H}_{k} \mathbf{P}_{k} \mathbf{P}_{k}^{H} \mathbf{H}_{k}^{H}
       \\ &\ \ \Big( \sum_{j\neq k} \mathbf{H}_{k} \mathbf{P}_{j} \mathbf{P}_{j}^{H} \mathbf{H}_{k}^{H}+\frac{\sigma_{k}^{2}}{P_{\max}}\sum_{i=1}^{K} \operatorname{Tr}( \mathbf{P}_{i} \mathbf{P}_{i}^{H}) \mathbf{I}\Big)^{-1}\Bigg).
\end{aligned}
\end{equation}
 
The relationship between problems \eqref{wsr_total_problem} and \eqref{unconstrained_WSR_problem} is established in the following proposition.

\begin{proposition}\label{proposition_unconstrained}\!\!\!\emph{:}
For any nontrivial stationary point $\{\mathbf{P}_{k}^{\star}\}$ of problem \eqref{wsr_total_problem}, there exists a stationary point $\{\mathbf{P}_{k}^{\ddagger}\}$ of the unconstrained problem \eqref{unconstrained_WSR_problem} such that  $\mathbf{P}_{k}^{\star}=\sqrt{\omega}\mathbf{P}_{k}^{\ddagger},\forall k$, where $\omega\triangleq\frac{P_{\max}}{\sum_{k=1}^{K}\operatorname{Tr}\left(\mathbf{P}_{k}^{\ddagger} (\mathbf{P}_{k}^{\ddagger})^{H}\right) }$ is a scaling factor, and vice versa.
\end{proposition}
\begin{proof}
See Appendix \ref{appendix_proposition_unconstrained}.
\end{proof}

Now let us turn our attention back to problem \eqref{eq: WSR-X}. In a similar way to Proposition \ref{proposition_optimal_with_equality} and problem \eqref{unconstrained_WSR_problem}, we can recast problem \eqref{eq: WSR-X} as the following unconstrained problem. 
\begin{equation}\label{unconstrained_WSR_problem_X}
\begin{aligned}
      & \max_{\mathbf{X}}  \  \  \sum_{k=1}^{K} \alpha_{k}\log \operatorname{det}\Bigg(\mathbf{I}+ \bar{\mathbf{H}}_{k} \mathbf{X}_{k} \mathbf{X}_{k}^{H} \bar{\mathbf{H}}_{k}^{H}
       \\ & \ \Big( \sum_{j\neq k} \bar{\mathbf{H}}_{k} \mathbf{X}_{j} \mathbf{X}_{j}^{H} \bar{\mathbf{H}}_{k}^{H}+\frac{\sigma_{k}^{2}}{P_{\max}}\sum_{i=1}^{K} \operatorname{Tr}( \bar{\mathbf{H}}\mathbf{X}_{i} \mathbf{X}_{i}^{H}) \mathbf{I}\Big)^{-1}\Bigg),
\end{aligned}
\end{equation}
where  $\bar{\mathbf{H}} \triangleq \mathbf{H H}^{H} \in \mathbb{C}^{N \times N}$, and $\bar{\mathbf{H}}_{k} \triangleq \mathbf{H}_{k}\mathbf{H}^{H} \in \mathbb{C}^{ N_{k} \times N}$ is the $k$-th submatrix of $\bar{\mathbf{H}}$. Furthermore, similar to Proposition \ref{proposition_unconstrained}, we have  the following Proposition.

\begin{proposition}\label{proposition_reduced_unconstrained}\!\!\emph{:}
For any nontrivial stationary point of problem \eqref{wsr_total_problem} $\mathbf{P}^{\star}=[\mathbf{P}_{1}^{\star},\mathbf{P}_{2}^{\star},\ldots,\mathbf{P}_{K}^{\star}]$, there exists a stationary point $\mathbf{X}^{\star}=[\mathbf{X}_{1}^{\star},\mathbf{X}_{2}^{\star},\ldots,\mathbf{X}_{K}^{\star}]$ of problem \eqref{unconstrained_WSR_problem_X}, such that
 $\mathbf{P}_{k}^{\star}=\sqrt{\beta}\mathbf{H}^{H}\mathbf{X}_{k}^{\star},\forall k$, where $\beta=\frac{P_{\max}}{\sum_{k=1}^{K}\operatorname{Tr}\left(\bar{\mathbf{H}}\mathbf{X}_{k}^{\star} (\mathbf{X}_{k}^{\star})^{H}\right)}$ is a scaling factor, and vice versa.
\end{proposition}
\begin{proof}
The proof follows immediately from Proposition \ref{proposition_column_space} and Proposition \ref{proposition_unconstrained}.
\end{proof}

Proposition \ref{proposition_reduced_unconstrained} implies that problem \eqref{wsr_total_problem} can be solved by solving \eqref{unconstrained_WSR_problem_X}. By applying Lemma \ref{lemma_WMMSE}, we have Theorem \ref{thm_transform}.
\begin{theorem}\label{thm_transform}\!\!\emph{:}
Problem \eqref{unconstrained_WSR_problem_X} is equivalent to the unconstrained weighted sum-MSE minimization problem \eqref{transform_problem_X} shown below.
\begin{equation}\label{transform_problem_X}
    \min_{\mathbf{U},\mathbf{W}, \mathbf{X}}~~ \sum_{k=1}^{K}\alpha_{k}\left(\operatorname{Tr}\left(\mathbf{W}_{k} \mathbf{E}_{k}^{\prime}\right)-\log \operatorname{det}\left(\mathbf{W}_{k}\right)\right),
\end{equation}
where $\mathbf{U}=\{\mathbf{U}_{k}\}_{k=1}^{K}$ and $\mathbf{W}=\{\mathbf{W}_{k}\}_{k=1}^{K}$ are introduced auxiliary variables, and  $\mathbf{E}_{k}^{\prime}$ is defined by
\begin{equation}\label{eq_Ek_prime}
    \begin{aligned}
        \mathbf{E}_{k}^{\prime}  \triangleq
                       & \left(\mathbf{I}-\mathbf{U}_{k}^{H} \bar{\mathbf{H}}_{k} \mathbf{X}_{k}\right)\left(\mathbf{I}-\mathbf{U}_{k}^{H} \bar{\mathbf{H}}_{k} \mathbf{X}_{k}\right)^{H}\\
&+\mathbf{U}_k^H\mathbf{N}_k^{\prime}\mathbf{U}_k
    \end{aligned}
\end{equation}
with 
$\mathbf{N}_k^{\prime}=\sum_{j\neq k} \bar{\mathbf{H}}_{k} \mathbf{X}_{j} \mathbf{X}_{j}^{H} \bar{\mathbf{H}}_{k}^{H}+\frac{\sigma_{k}^{2}}{P_{\max}}\sum_{i=1}^{K} \operatorname{Tr}( \bar{\mathbf{H}}\mathbf{X}_{i} \mathbf{X}_{i}^{H}) \mathbf{I}$.
\end{theorem}

\subsection{The Proposed R-WMMSE Algorithm}
Now we are ready to propose the R-WMMSE algorithm by applying BCD to problem \eqref{transform_problem_X}. Note that the objective function of problem \eqref{transform_problem_X} is a convex function of each block variable $\mathbf{U}$, $\mathbf{W}$, and $\mathbf{X}$, respectively. Hence, in the BCD applied to \eqref{transform_problem_X}, i.e., solve for one block variable in each iteration while fixing the others, the subproblem with respect to each block variable can be globally solved in closed-form by the first-order optimality, leading to the R-WMMSE algorithm with the following  closed-form updates in each iteration:
 \begin{enumerate}
\item {\bf Update $\mathbf{U}$} by
\begin{equation}\label{U_update}
\begin{aligned}
\mathbf{U}_{k}=\Bigg(&\sum_{i=1}^{K} \frac{\sigma_{k}^{2}}{P_{\max }} \operatorname{Tr}\left(\bar{\mathbf{H}} \mathbf{X}_{i} \mathbf{X}_{i}^{H}\right) \mathbf{I}\\
&+\sum_{j=1}^{K} \bar{\mathbf{H}}_{k} \mathbf{X}_{j} \mathbf{X}_{j}^{H} \bar{\mathbf{H}}_{k}^{H}\Bigg)^{-1}
        \bar{\mathbf{H}}_{k} \mathbf{X}_{k},\; \forall\; k.
\end{aligned}
\end{equation}
\item {\bf Update $\mathbf{W}$} by
\begin{equation}\label{W_update}
 \mathbf{W}_{k}=\left(\mathbf{I}-\mathbf{U}_{k}^{H} \bar{\mathbf{H}}_{k} \mathbf{X}_{k}\right)^{-1},\; \forall\; k.
\end{equation}
 \item {\bf Update $\mathbf{X}$} by 
\begin{equation}\label{Xk_update}
    \begin{aligned}
        \mathbf{X}_{k}=\Bigg(&\sum_{i=1}^{K} \frac{\sigma_{i}^{2}}{P_{\max }} \alpha_{i} \operatorname{Tr}\left(\mathbf{M}_{i}\right) \bar{\mathbf{H}}\\
        &+\sum_{j=1}^{K} \alpha_{j} \bar{\mathbf{H}}_{j}^{H} \mathbf{M}_{j} \bar{\mathbf{H}}_{j}\Bigg)^{-1}
        \alpha_{k} \bar{\mathbf{H}}_{k}^{H} \mathbf{U}_{k} \mathbf{W}_{k},\; \forall\; k,
    \end{aligned}
\end{equation}
where $\mathbf{M}_{k} \triangleq \mathbf{U}_{k} \mathbf{W}_{k} \mathbf{U}_{k}^{H}$.
 \end{enumerate}

The step to update $\mathbf{X}$ requires an $N$-dimensional matrix inversion operation whose complexity is independent of $M$. However, we show below that the computational complexity can be further reduced when $D< N$. Specifically, let us define $\eta\triangleq \sum_{i=1}^{K} \frac{\sigma_{i}^{2}}{P_{\max }} \alpha_{i}\operatorname{Tr}\left(\mathbf{M}_{i}\right)$, $\hat{\mathbf{W}}\triangleq\text{blkdiag}\left(\alpha_{1}\mathbf{W}_1,\ldots,\alpha_{K}\mathbf{W}_K\right)$, $\hat{\mathbf{U}}\triangleq\text{blkdiag}\left(\mathbf{U}_1,\ldots,\mathbf{U}_K\right)$. Then, we can write \eqref{Xk_update} in a more compact form as follows
\begin{equation}\label{Xk_update_low_dimensional}
    \begin{aligned}
        \mathbf{X}&=\left(\eta\bar{\mathbf{H}}+ \bar{\mathbf{H}} \hat{\mathbf{U}} \hat{\mathbf{W}}  \hat{\mathbf{U}}^{H} \bar{\mathbf{H}} \right)^{-1}
       \bar{\mathbf{H}} \hat{\mathbf{U}} \hat{\mathbf{W}}\\
        &=\left(\eta\mathbf{I}+  \hat{\mathbf{U}} \hat{\mathbf{W}} \hat{\mathbf{U}}^{H} \bar{\mathbf{H}} \right)^{-1}
       \hat{\mathbf{U}} \hat{\mathbf{W}}\\
        &=\hat{\mathbf{U}} \hat{\mathbf{W}}\left(\eta\mathbf{I}+  \hat{\mathbf{U}}^{H} \bar{\mathbf{H}} \hat{\mathbf{U}} \hat{\mathbf{W}}  \right)^{-1}\\
         &=\hat{\mathbf{U}} \left(\eta\hat{\mathbf{W}}^{-1}+  \hat{\mathbf{U}}^{H} \bar{\mathbf{H}}\hat{\mathbf{U}} \right)^{-1},
    \end{aligned}
\end{equation}
where the second equality follows from the assumption of full row rank $\mathbf{H}$, the third equality is due to the identity $\left(\mathbf{I}+\mathbf{A}\mathbf{B}\right)^{-1}\mathbf{A}=\mathbf{A}\left(\mathbf{I}+\mathbf{B}\mathbf{A}\right)^{-1}$, and the last equality holds because $\hat{\mathbf{W}}$ is invertible. As compared to \eqref{Xk_update}, \eqref{Xk_update_low_dimensional} involves only a $D$-dimensional matrix inversion operation.

Using \eqref{Xk_update_low_dimensional} instead of \eqref{Xk_update}, the R-WMMSE algorithm is summarized in Algorithm \ref{alg_R-WMMSE}. Once the convergence is reached, the final precoders are obtained by $\mathbf{P}_k=\sqrt{\beta}\mathbf{H}^{H}\mathbf{X}_k, \forall k$ with $\beta$ defined in Proposition \ref{proposition_reduced_unconstrained}. 
\begin{algorithm}[!tb]
    \renewcommand{\algorithmicrequire}{\textbf{Initialization:}}
    \renewcommand{\algorithmicensure}{\textbf{Output:}}
    \caption{The Proposed R-WMMSE Algorithm}  \label{alg_R-WMMSE}
    \begin{algorithmic}[1]
    \REQUIRE Compute $\bar{\mathbf{H}}=\mathbf{H}\mathbf{H}^H$. Initialize $\mathbf{X}$ with $\sum_{k=1}^{K}\operatorname{Tr}\left(\bar{\mathbf{H}} \mathbf{X}_{k} \mathbf{X}_{k}^{H}\right)\leq P_{\max }$ and $\mathbf{W}_k=\mathbf{I},\forall k$. Set the tolerance of accuracy $\epsilon$.
    \REPEAT
    \STATE $\mathbf{W}'_{k} = \mathbf{W}_{k},\; \forall\; k$;\\
    \STATE Update $\mathbf{U}_{k}$'s by \eqref{U_update};
    \STATE  Update $\mathbf{W}_{k}$'s by \eqref{W_update};
     \STATE Update $\mathbf{X}$ by \eqref{Xk_update_low_dimensional};
    \UNTIL $|\sum_k \alpha_k\log\det(\mathbf{W}_k)-\sum_k \alpha_k\log\det(\mathbf{W}'_k)|\leq \epsilon$.
    \ENSURE $\mathbf{P}_k{=}\sqrt{\beta}\mathbf{H}^H\mathbf{X}_k,\forall k$ such that the SPC with equality.
    \end{algorithmic}
\end{algorithm}

Interestingly, computing $\bar{\mathbf{H}}$ has complexity $\mathcal{O}(MN^2)$, which is linear in $M$. Moreover, once $\bar{\mathbf{H}}$ is determined, each iteration of R-WMMSE is independent of $M$. Furthermore, the computational complexity of the proposed R-WMMSE algorithm is dominated by the matrix inversion operation for update of $\mathbf{X}$   (i.e., line 5 of Algorithm \ref{alg_R-WMMSE}), which is in the order of $\mathcal{O}(D^3)$, much lower than the complexity $\mathcal{O}(M^3)$ of the classical WMMSE algorithm \cite{shi2011iteratively}. Therefore, due to the computation of $\bar{\mathbf{H}}$, the proposed R-WMMSE has a complexity of $\mathcal{O}(MN^2)$ in total and thus is more suitable than WMMSE for real-time implementation in massive MU-MIMO systems. 

\begin{remark}\!\!\emph{:}
It is seen that, there is no apparent complexity reduction by using \eqref{Xk_update_low_dimensional} instead of \eqref{Xk_update} when $N=D$. In this case, we propose to update $\mathbf{X}_k$ in \eqref{Xk_update} with a computationally efficient operation of matrix inversion. Specifically, the efficient inversion is made based on the Woodbury matrix identity\cite{petersen2008matrix}, namely
\begin{equation}\label{Woodbury}
\begin{aligned}
&\left(\mathbf{A}+\mathbf{C B C}^{H}\right)^{-1}\\
=&\mathbf{A}^{-1}-\mathbf{A}^{-1} \mathbf{C}\left(\mathbf{B}^{-1}
+\mathbf{C}^{H} \mathbf{A}^{-1} \mathbf{C}\right)^{-1} \mathbf{C}^{H} \mathbf{A}^{-1}.
\end{aligned}
\end{equation}
By applying the Woodbury identity, the inversion operation required by the update of $\mathbf{X}_k$ in \eqref{Xk_update} can be realized by the following recursion ($l=1,2,\ldots K$): 
\begin{equation}
\left\{\begin{array}{l}\label{simplified_Xk}
\mathbf{A}_{l}^{-1}=\left(\mathbf{A}_{l-1}^{-1}+\mathbf{C}_{l} \mathbf{B}_{l} \mathbf{C}_{l}^{H}\right)^{-1}, \\
\mathbf{C}_{l}=\bar{\mathbf{H}}_{l}^{H}  \mathbf{U}_{l}, \\
\mathbf{B}_{l}=\alpha_l\mathbf{W}_{l},
\end{array}\right.
\end{equation}
\begin{equation}
\mathbf{A}_{0}^{-1}=\frac{1}{\sum_{i=1}^{K}\frac{\sigma_{i}^{2}}{P_{\max }} \alpha_{i} \operatorname{Tr}\left(\mathbf{M}_{i} \right)}\bar{\mathbf{H}}^{-1}.
\end{equation}
Clearly, $\mathbf{A}_{K}^{-1}$ gives the desired matrix inversion in \eqref{Xk_update}, i.e., $(\sum_{i=1}^{K} \frac{\sigma_{i}^{2}}{P_{\max }} \alpha_{i} \operatorname{Tr}\left(\mathbf{M}_{i}\right) \bar{\mathbf{H}}+\sum_{j=1}^{K} \alpha_{j} \bar{\mathbf{H}}_{j}^{H} \mathbf{M}_{j} \bar{\mathbf{H}}_{j})^{-1}$. It is emphasized that, by the above recursion, only one large-scale matrix inversion operation, i.e, $\bar{\mathbf{H}}^{-1}\in \mathbb{C}^{N\times N}$, is needed in the entire iterative process of the R-WMMSE algorithm. 
\end{remark}

Finally, we establish the convergence of the R-WMMSE algorithm in  Theorem \ref{th_convergence_result}.

\begin{theorem}\label{th_convergence_result}
\emph{(Convergence Results for the R-WMMSE Algorithm):}
Any limit point $\left(\mathbf{U}^{\star},\mathbf{W}^{\star},\mathbf{X}^{\star}\right)$ of the iterative sequence generated by the R-WMMSE algorithm is a stationary point of problem
\eqref{transform_problem_X}, and the corresponding $\mathbf{P}^{\star}=\sqrt{\beta}\mathbf{H}^{H}\mathbf{X}^{\star}$ is a nontrivial stationary point of problem \eqref{wsr_total_problem} with $\beta$ defined in Proposition \ref{proposition_reduced_unconstrained}.
\end{theorem}
\begin{proof}
See Appendix \ref{appendix_theorem2}.
\end{proof}


 So far, we have proposed a linear-complexity R-WMMSE algorithm with convergence guaranteed for the WSR maximization problem under SPC. A natural question is how it can be extended to the PAPCs case.
\begin{remark}\emph{(Straightforward Extension to the PAPCs Case):}
The R-WMMSE algorithm can be used to generate a feasible solution for the WSR maximization problem with PAPCs. That is, given the R-WMMSE solution $\mathbf{P}_k^{\text{SPC}}$'s, we can obtain a feasible solution for the PAPCs case by simply normalizing $\mathbf{P}_k^{\text{SPC}}$'s to satisfy the PAPCs, i.e.,
\begin{equation}\label{eq_normlaized_WMMSE}
 \mathbf{P}^{\text{normalized}}_k=\mathbf{P}_k^{\text{SPC}}\times \min_{m \in \{1,\ldots,M\}} \frac{\sqrt{P_m}}{\|\mathbf{p}_{m}\|_2},\ \forall k,   
\end{equation} 
where $\mathbf{p}_{m}^{H}$ is the $m$-th row vector of $\mathbf{P}^{\text{SPC}}$. However, this feasible solution may be far from optimum for the WSR maximization problem with PAPCs.
\end{remark}

\section{The Proposed PAPC-WMMSE for the PAPCs Case}
In this section, we develop an efficient iterative solution to the WSR maximization problem with PAPCs by directly applying the WMMSE framework to problem \eqref{wsr_PAPC_problem}.

Specifically, by applying Lemma \ref{lemma_WMMSE}, we obtain the following equivalent problem of problem \eqref{wsr_PAPC_problem}
\begin{equation}\label{transform_problem_PAPC}
\begin{aligned}
\min _{\mathbf{W},\mathbf{U},\mathbf{P}}\  \ & \sum_{k=1}^{K} \alpha_{k} \left(\operatorname{Tr}\left(\mathbf{W}_{k}\mathbf{E}_{k}\right) - \log\det\left(\mathbf{W}_{k}\right) \right) \\
\text { s.t.~~}\  \ & \sum_{k=1}^{K} \left[\mathbf{P}_{k} \mathbf{P}_{k}^{H}\right]_{m,m} \leq P_{m},\ \forall m, 
\end{aligned}
\end{equation}
where $\mathbf{W}=\{\mathbf{W}_{k}\}_{k=1}^{K}$, $\mathbf{U}=\{\mathbf{U}_{k}\}_{k=1}^{K}$ are auxiliary variables, and  $\{\mathbf{E}_{k}\}_{k=1}^{K}$ are all exactly the same as defined in \eqref{MSE_matrix_WMMSE}. Similarly, we can directly apply BCD to \eqref{transform_problem_PAPC} by dividing the variables into three block variables, i.e, $\mathbf{U}$, $\mathbf{W}$, and $\mathbf{P}$. However, in this way, we need to solve a convex QCQP for updating $\mathbf{P}$ using some sophisticated optimization methods with high complexity (e.g., the interior-point method for this problem requires a computational complexity of $\mathcal{O}(M^{3.5}D^3)$\cite{wang2014outage}). Therefore, such a straightforward way of applying BCD to \eqref{transform_problem_PAPC} is not suitable for the PAPCs case. Our goal is to derive simple closed-form updates for all variables with linear complexity.  


Actually, in terms of the PAPCs, we could naturally view each row of $\mathbf{P}$ or each column of  $\mathbf{P}^H$ as a block variable. That is, we can apply BCD to \eqref{transform_problem_PAPC} with block variables  $\mathbf{U}$, $\mathbf{W}$, and $\mathbf{p}_m$, $m=1,2,\ldots,M$, where $\mathbf{p}_m$ denotes the $m$-th column of  $\mathbf{P}^H$. As a result, we can update $\mathbf{U}$ and $\mathbf{W}$ in closed-form as the WMMSE algorithm, while the main difficulty lies in the update of each column of $\mathbf{P}^H$, i.e., solving the following problem:
\begin{equation}\label{P_PAPC_BCD}
    \begin{aligned}
       \min_{\mathbf{p}_m}  \  \ &  
      \sum_{k=1}^{K}  -2\operatorname{Tr}\left(\Re e\left( \alpha_{k}\mathbf{W}_{k}\mathbf{U}_{k}^{H}\sum_{m=1}^{M}\left(\mathbf{h}_{k,m} \mathbf{p}_{k,m}^{H}\right)\right)\right)\\
       & + \sum_{k=1}^{K}\operatorname{Tr}\Bigg(\sum_{j=1}^{K}\alpha_{j} \sum_{m=1}^{M}\left(\mathbf{p}_{k,m}\mathbf{h}_{j,m}^{H}\right)\mathbf{U}_{j}\mathbf{W}_{j}\\
       &~~\times\mathbf{U}_{j}^{H}\sum_{m=1}^{M}\left(\mathbf{h}_{j,m} \mathbf{p}_{k,m}^{H}\right)\Bigg)      \\
        \text{s.t.~~}&          \sum_{k=1}^{K} \|\mathbf{p}_{k,m}\|_2^2\leq P_{m},
    \end{aligned}
\end{equation}
where $\mathbf{p}_{m}=\left(\mathbf{p}_{1,m}^{T},\ldots,\mathbf{p}_{K,m}^{T}\right)^{T}$ with $\mathbf{p}_{k,m}$ being the $m$-th column of $\mathbf{P}_{k}^H$. Note that we have used in the above the fact  $\mathbf{H}_j\mathbf{P}_k=\sum_{m=1}^{M}\mathbf{h}_{j,m} \mathbf{p}_{k,m}^{H}, \forall j,k$, with $\mathbf{h}_{k,m}$ being the $m$-th column of $\mathbf{H}_k$.

Problem \eqref{P_PAPC_BCD} is a convex quadratic optimization problem with a single constraint, which is generally solved via the Bisection method. However, by further exploring the problem structure, we find that the quadratic term with respect to $\mathbf{p}_{m}$ can be expressed as a scaled Euclidean norm. This benign {\em scaled-norm property} facilitates a closed-form optimal solution. Specifically, letting $\mathbf{A}\triangleq\text{blkdiag}\left(\alpha_{1}\mathbf{U}_{1}\mathbf{W}_{1}\mathbf{U}_{1}^{H},\ldots,\alpha_{K}\mathbf{U}_{K}\mathbf{W}_{K}\mathbf{U}_{K}^{H}\right)\in \mathbb{C}^{N \times N}$, $\mathbf{B}\triangleq\text{blkdiag}\left(\alpha_{1}\mathbf{W}_{1}\mathbf{U}_{1}^{H},\ldots,\alpha_{K}\mathbf{W}_{K}\mathbf{U}_{K}^{H}\right)\in \mathbb{C}^{D \times N}$, and  $\mathbf{h}_{m}\triangleq\left(\mathbf{h}_{1,m}^{T},\ldots,\mathbf{h}_{K,m}^{T}\right)^{T}\in \mathbb{C}^{N\times 1}$, problem \eqref{P_PAPC_BCD} can be recast as
\begin{equation}\label{P_PAPC_final}
    \begin{aligned}
        \min_{\mathbf{p}_{m}} ~~
        &a_{m} \|\mathbf{p}_{m}\|_2^{2}+2\Re e\left(\mathbf{b}_{m}^H\mathbf{p}_{m}
        \right)
        \\
        \text{ s.t. } ~~    &  \|\mathbf{p}_{m}\|_2^2\leq P_{m},
    \end{aligned}
\end{equation}
where 
\begin{equation}\label{eq_a_m}
a_{m}\triangleq\mathbf{h}_{m}^{H}\mathbf{A} \mathbf{h}_{m},\end{equation}
and
\begin{equation}\label{eq_b_km}
\mathbf{b}_{m}\triangleq -\mathbf{B}\mathbf{h}_{m}+\sum_{l \neq m}\mathbf{p}_{l}\mathbf{h}_{l}^{H}\mathbf{A} \mathbf{h}_{m}. \end{equation}
It is readily seen that \eqref{P_PAPC_final}  is essentially  the problem of projection of the vector $-\frac{1}{a_m} \mathbf{b}_m$ onto the Euclidean ball $\|\mathbf{p}_{m}\|_2\leq \sqrt{P_{m}}$, which admits a closed-form solution as follows.
\begin{equation}\label{P_PAPC_solution_brief}
\mathbf{p}_{m}=-\mathbf{b}_{m}\times \text{min}\left( \frac{1}{a_{m}},\frac{\sqrt{P_{m}}}{ \|\mathbf{b}_{m}\|_2} \right).
\end{equation}

Apparently, computing $\mathbf{b}_{m}$ has a complexity of $\mathcal{O}\left(M\right)$ due to the summation over $l=1,2,\ldots, m-1, m+1,\ldots, M$ in the second term of  \eqref{eq_b_km}. As a result, computing all $\mathbf{b}_{m}$, $m=1,2,\ldots,M$ in each iteration may require a complexity of $\mathcal{O}\left(M^2\right)$. However, observing the special structure of $\sum_{l\neq m} \mathbf{p}_{l}\mathbf{h}_l^H$ which can be expressed as $\sum_{l=1}^M \mathbf{p}_{l}\mathbf{h}_l^H{-}\mathbf{p}_{m}\mathbf{h}_m^H$, we can recursively update $\mathbf{C} \triangleq\sum_{l\neq m} \mathbf{p}_{l}\mathbf{h}_l^H$ as a whole while computing each $\mathbf{b}_{m}$ via \eqref{eq_b_km}. In such a way, we obtain a linear-complexity algorithm for the WSR maximization problem with PAPCs (see Algorithm \ref{alg_PAPC-WMMSE}), termed \emph{PAPC-WMMSE}. In the PAPC-WMMSE algorithm, we iteratively update $\mathbf{U}$, $\mathbf{W}$, and $\mathbf{p}_{m}$, $m=1,2,\ldots,M$ until convergence. By a similar proof as in \cite{shi2011iteratively},  we have the following  convergence result stated in Theorem \ref{th_PAPC_WMMSE_convergence}.

\begin{theorem}\label{th_PAPC_WMMSE_convergence}
\emph{(Convergence Results of the PAPC-WMMSE Algorithm):}
Any limit point $\left(\mathbf{U}^{\star},\mathbf{W}^{\star},\mathbf{P}^{\star}\right)$ of the iterative sequence generated by the PAPC-WMMSE algorithm is a stationary point of problem
\eqref{transform_problem_PAPC}, and the corresponding $\mathbf{P}^{\star}$ is a stationary point of problem \eqref{wsr_PAPC_problem}.
\end{theorem}


\begin{algorithm}[!tb]
    \renewcommand{\algorithmicrequire}{\textbf{Initialization:}}
    \renewcommand{\algorithmicensure}{\textbf{Output:}}
    \caption{Proposed PAPC-WMMSE Precoding}  \label{alg_PAPC-WMMSE}
    \begin{algorithmic}[1]
    \REQUIRE Initialize $\mathbf{P}$ satisfying PAPCs and $\mathbf{W}_k=\mathbf{I},\forall k$. Set the tolerance of accuracy $\epsilon$.
    \REPEAT
    \STATE $\mathbf{W}'_{k} = \mathbf{W}_{k},\; \forall\; k$\\
    \STATE $\mathbf{U}_{k} = \left(\sum_{j=1}^K \mathbf{H}_{k}\mathbf{P}_{j}\mathbf{P}_{j}^{H}\mathbf{H}_{k}^{H}+\sigma_{k}^{2} \mathbf{I} \right)^{-1}\mathbf{H}_{k}\mathbf{P}_{k},\; \forall\; k$;
    \STATE $\mathbf{W}_{k} = \left(\mathbf{I}-\mathbf{U}_k^H\mathbf{H}_k\mathbf{P}_k \right)^{-1},\; \forall\; k$;
    \STATE $\mathbf{C} = \sum_{l =1 }^{M}\mathbf{p}_{l}\mathbf{h}_{l}^{H}$;
    \STATE \textbf{for} $m = 1:M$ \textbf{do}\\
    \STATE\quad $a_{m}=\mathbf{h}_{m}^{H}\mathbf{A} \mathbf{h}_{m}$;\\
    \STATE\quad ${\mathbf{C}}=\mathbf{C}-\mathbf{p}_{m}\mathbf{h}_{m}^{H}$;\\
\STATE\quad 
$\mathbf{b}_{m}= -\mathbf{B}\mathbf{h}_{m}+\mathbf{C}\mathbf{A} \mathbf{h}_{m}$;
    \STATE\quad
    $\mathbf{p}_{m}=-\mathbf{b}_{m}\times \text{min}\left( \frac{1}{a_{m}},\frac{\sqrt{P_{m}}}{ \|\mathbf{b}_{m}\|_2} \right)$;\\
    \STATE\quad
    $\mathbf{C}={\mathbf{C}}+\mathbf{p}_{m}\mathbf{h}_{m}^{H}$;\\
    \STATE\textbf{end for}\\
    \UNTIL $|\sum_k \alpha_k\log\det(\mathbf{W}_k)-\sum_k \alpha_k\log\det(\mathbf{W}'_k)|\leq \epsilon$.
        \ENSURE $\mathbf{P}=\left(\mathbf{p}_{1},\ldots,\mathbf{p}_{M}\right)^{H}$.
    \end{algorithmic}
\end{algorithm}

At last, let us analyze the complexity of the proposed PAPC-WMMSE. In  Algorithm \ref{alg_PAPC-WMMSE}, When $M\gg N\ge D\geq K$, it can be shown that the computational complexity of lines 3-5 is dominated by line 5, which has a complexity of $\mathcal{O}(MND)$. Moreover, because $\mathbf{A}$ is a block diagonal matrix, the computational complexity of lines 7-11 is dominated by line 9, which has a complexity of $\mathcal{O}(ND)$ due to the computation of $\mathbf{C}\mathbf{A}\mathbf{h}_m$. Therefore, we conclude that the computational complexity of each iteration of the proposed PAPC-WMMSE is $\mathcal{O}(MND)$ which is linear in the number of BS antennas and significantly less than the cubic complexity (i.e., $\mathcal{O}\left(M^3\right)$) of the existing algorithms  \cite{shi2008per,Mao2019rate,thomas2018hybrid}.

\section{Simulation Results}

\subsection{Simulation Setup}


 We consider a single-cell massive MU-MIMO system consisting of a BS with $M$ antennas and $K$ users each equipped with $N_k=4$ antennas to receive $D_k=2$ or $4$ data streams. The sum power budget of the BS for the SPC case is set to be $P_{\text{max}}=10$ [W], while the maximum per-antenna transmit power for PAPCs case is set to be $P_{\text{max}}/M$.
The channel matrix $\mathbf{H}$ is generated from the circularly-symmetric standard complex normal distribution with pathloss between the users and the BS. The pathloss is set to be $128.1+37.6\log_{10}\left(\omega\right)$[dB]\cite{dahrouj2010coordinated}, where $\omega$ is the distance between the user and the BS taking range in $0.1\sim0.3$ km. The noise power is set to be equal for all users and is given by $\sigma_{k}^{2}=10^{\frac{1}{K} \sum_{k} \log _{10} \frac{1}{N_{k}}  \|\mathbf{H}_{k}\|_{F}^{2} }\times 10^{-\frac{\mathrm{SNR}}{10}}$, where signal-to-noise ratio (SNR) is the
average receive SNR for all users when no precoding is used. The priority weights $\alpha_k,\forall k$ of the users are set to be equal. Our simulation results are averaged over 100 randomly generated channel realizations.

All computations below are performed using an AMD Ryzen 7 5800H with Radeon Graphics 3.20 GHz, 16 GB Memory (RAM), Windows 10 (64 b) operating system, and Matlab R2020a environment.

\subsection{R-WMMSE Performance Evaluation for the SPC Case}

This subsection provides simulation results evaluating the performance of the proposed R-WMMSE algorithm for the SPC case. We compare our method with other baselines, including the WMMSE algorithm in \cite{shi2011iteratively} and the ZF precoding method. The ZF precoder serves as the initial point of the other two methods.

\begin{figure}[htbp]
    \centering
    \subfloat[$M=64$, $K=12$, $D_k=2$, $10$ dB.]{\includegraphics[width=0.45\textwidth]{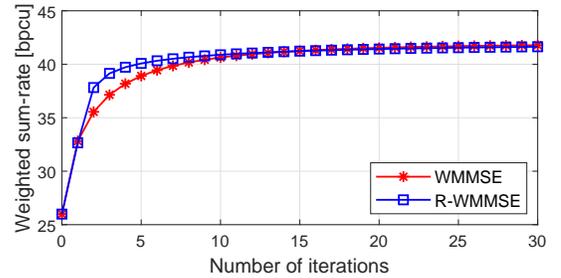}}\\
    \subfloat[$M=128$, $K=16$, $D_k=4$, $0$ dB.]{\includegraphics[width=0.45\textwidth]{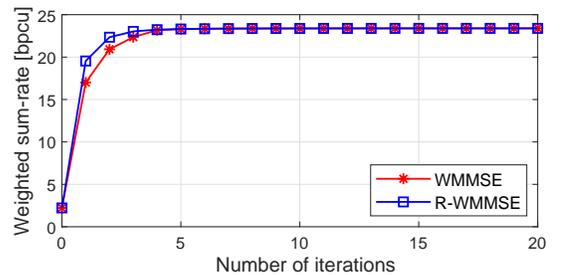}}\\
    \caption{Convergence of the proposed R-WMMSE algorithm and the WMMSE algorithm.}
    \label{fig:th_convergence_results}
\end{figure}

First, we examine the convergence performance of the proposed R-WMMSE algorithm and the WMMSE algorithm in Fig. \ref{fig:th_convergence_results}(a) and (b). The weighted sum-rate is measured by bits per channel use (bpcu). These plots show that the proposed R-WMMSE algorithm and the WMMSE algorithm converge to the same WSR value, conforming to the previous theoretical finding. Furthermore, it is observed that starting from the same initial point, the R-WMMSE algorithm could often have better convergence performance (especially in the first several iterations) than the WMMSE algorithm because the R-WMMSE algorithm relaxes the SPC during iterations. 

\begin{figure}[htbp]
	\centering
	\subfloat[$M=128$, $D_k=2$.]{\label{fig:time_results_vs_K}\includegraphics[width=0.45\textwidth]{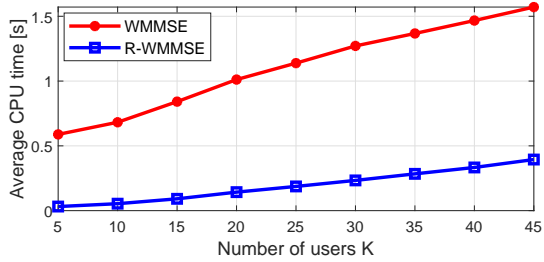}}\\
	\subfloat[ $K=16$, $D_k=4$.]{\label{fig:time_results_vs_M}\includegraphics[width=0.45\textwidth]{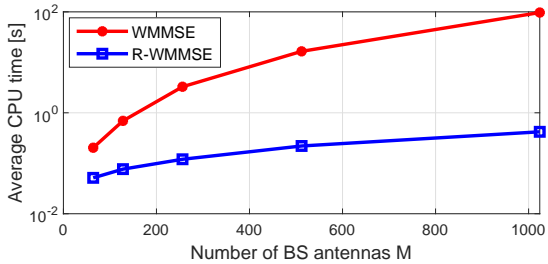}}\\	
	\caption{Average CPU time versus: (a) the number of users, and (b) the number
of BS antennas.}\label{fig:SPC_convergence}
\end{figure}

Second, we compare the proposed R-WMMSE algorithm with the classical WMMSE algorithm in terms of the average CPU execution time under different numbers of users $K$ and different numbers of BS antennas $M$. As can be seen in Fig. \ref{fig:SPC_convergence}(a), when $M=128$ is fixed, the average execution time of the R-WMMSE is about $20\%$ of that of the WMMSE algorithm. Moreover, the CPU time of both algorithms increases slowly with $K$ increasing. As shown in Fig. \ref{fig:SPC_convergence}(b), when fixing $K=16$, the CPU time\footnote{Note that we take the logarithm of the CPU time here since the CPU time gap between the two algorithms is too large when $M$ is large.} of the WMMSE algorithm increases sharply as $M$ increases. In contrast, the CPU time of the proposed R-WMMSE algorithm increases only slightly. In particular, when $M=1024$, the WMMSE algorithm will take $97$ seconds to converge, while the R-WMMSE algorithm only takes $0.4$ seconds. In other words, for the extremely large antenna array case with $M>1000$, the R-WMMSE algorithm performs 200+ times faster than the WMMSE algorithm. These phenomena are consistent with the previous complexity analysis that the R-WMMSE algorithm and the WMMSE algorithm have linear and cubic complexity in $M$, respectively.

\begin{figure}[htbp]
    \centering
    \includegraphics[width=0.5\textwidth]{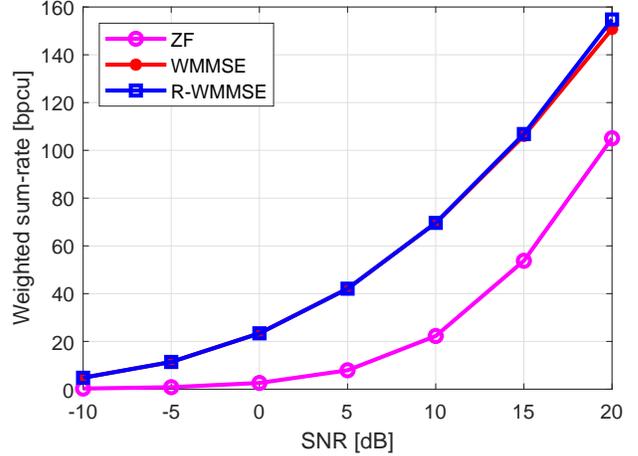}
    \caption{WSR performance versus SNR ($M=128,K=16,D_k=4$).}
    \label{fig:wsr_vs_snr} 
\end{figure}

Finally, we present the weighted sum-rate performance of the proposed R-WMMSE algorithm and other baselines versus SNR under the setting: $M=128,K=16$, and $D_k=4$. As shown in Fig. \ref{fig:wsr_vs_snr}, the proposed R-WMMSE algorithm yields almost the same performance as the WMMSE algorithm, but significantly outperforms the ZF algorithm under different SNRs. The main reason for the gap between the WMMSE/R-WMMSE and ZF is that the power control is not optimized in the latter.

\subsection{PAPC-WMMSE Performance Evaluation for the PAPCs case}

This subsection provides numerical results of the proposed PAPC-WMMSE algorithm for the PAPCs case. We compare our method with the baselines, including normalized WMMSE, normalized ZF, and HBF in \cite{thomas2018hybrid}. Here, the normalization is performed according to Eq. \eqref{eq_normlaized_WMMSE}, and the normalized ZF algorithm serves as the initial point of the iterative methods.
\begin{figure}[htbp]
	\centering
	\subfloat[$M=64$, $K=12$, $D_k=2$, $10$ dB.]{\label{fig:papc_convergence_a}\includegraphics[width=0.45\textwidth]{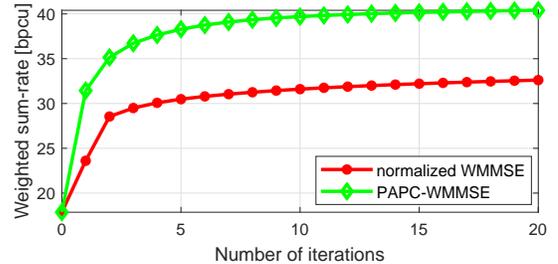}}\\
	\subfloat[$M=128$, $K=16$, $D_k=4$, $0$ dB.]{\label{fig:papc_convergence_b}\includegraphics[width=0.45\textwidth]{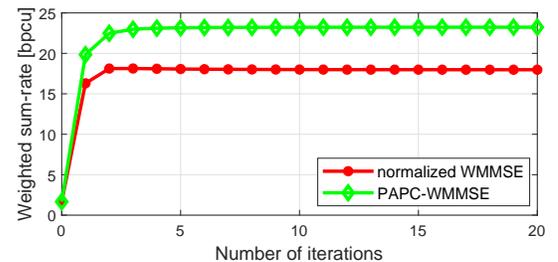}}\\	
	\caption{Convergence of the proposed PAPC-WMMSE algorithm.}\label{fig:papc_convergence}
\end{figure}

First, we examine the convergence performance of the proposed PAPC-WMMSE algorithm and the normalized WMMSE algorithm. As can be seen in Fig. \ref{fig:papc_convergence}(a) and (b), the proposed PAPC-WMMSE algorithm and the normalized WMMSE algorithm both converge within 20 iterations for case (a), while only 5 iterations for case (b). Furthermore, we find that the proposed PAPC-WMMSE algorithm has a much better WSR performance than the normalized WMMSE algorithm.

\begin{figure}[htbp]
	\centering
	\subfloat[$M=128$,  $D_k=2$.]{\label{fig:papc_CPU_time_K}\includegraphics[width=0.45\textwidth]{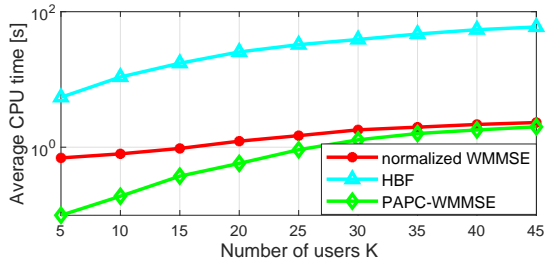}}\\
	\subfloat[$K=12$, $D_k=2$.]{\label{fig:papc_CPU_time_M}\includegraphics[width=0.45\textwidth]{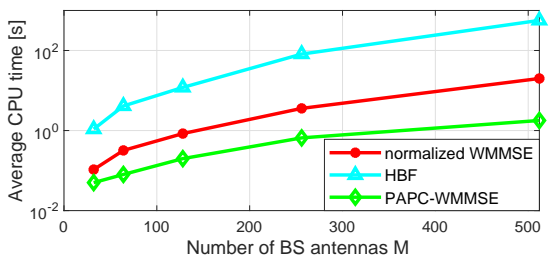}}\\	
	\caption{Average CPU time versus: (a) the number of users, and (b) the number
of BS antennas.}\label{fig:papc_CPU_time}
\end{figure}

Second, in terms of the average CPU execution time, we compare the proposed PAPC-WMMSE algorithm with the baselines under different numbers of users $K$ and different numbers of BS antennas $M$. As shown in Fig. \ref{fig:papc_CPU_time}(a), when $M=128$ is fixed, the average CPU time of the PAPC-WMMSE algorithm is less than the normalized WMMSE algorithm, especially when $K$ is small. As shown in Fig. \ref{fig:papc_CPU_time}(b), when fixing $K=12$, as $M$ increases, the HBF algorithm and the normalized WMMSE algorithm will take significantly more time than the proposed PAPC-WMMSE algorithm. In particular, when $M=512$, the HBF and the normalized WMMSE algorithm will take $565$ and $20$ seconds to converge, respectively. In contrast, the PAPC-WMMSE algorithm only takes $1.7$ seconds. In conclusion, the PAPC-WMMSE algorithm performs significantly faster than the HBF and the normalized WMMSE algorithm. These phenomena are consistent with the previous complexity analysis that the PAPC-WMMSE algorithm has linear complexity in $M$. 

\begin{figure}[htbp]
    \centering
    \includegraphics[width=0.5\textwidth]{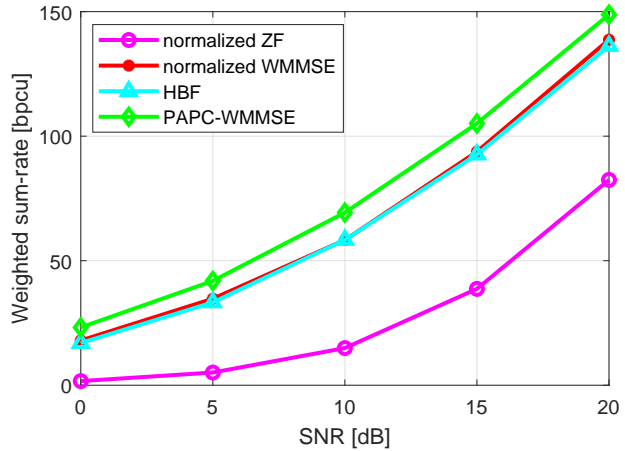}
    \caption{WSR performance versus SNR under PAPCs ($M=128,K=16,D_k=4$).}
    \label{fig:papc_wsr_vs_snr}
\end{figure}

Finally, we evaluate the WSR performance of the proposed PAPC-WMMSE algorithm and other baselines versus SNR under the setting: $M=128,K=16$, and $D_k=4$. As shown in Fig. \ref{fig:papc_wsr_vs_snr}, the proposed PAPC-WMMSE algorithm outperforms the HBF algorithm in \cite{thomas2018hybrid} and the normalized WMMSE algorithm, while the normalized ZF algorithm achieves the worst WSR performance since it is very heuristic without considering the objective of WSR maximization.

\section{Conclusions}
This paper has investigated the WSR maximization problems of massive MU-MIMO systems under SPC and PAPCs. The existing optimization-based algorithms for both problems suffer from the cubic complexity issue. This paper has proposed two linear-complexity algorithms (i.e., R-WMMSE and PAPC-WMMSE) for the WSR maximization problems with SPC and PAPCs, respectively, and established their convergence. Notably, the per-iteration complexity of R-WMMSE is independent of the number of BS antennas, making it attractive for extremely-large-antenna arrays. 

Finally, it is worth to mention that, in spite of focusing on the single-cell MU-MIMO system, the ideas behind the R-WMMSE algorithm and the PAPC-WMMSE algorithm can be easily generalized to the multi-cell MU-MIMO systems, since the \emph{low-dimensional subspace} property may be very useful for \emph{bandwidth-constrained} decentralized transceiver design in cell-free massive MIMO systems, decentralized baseband processing, multi-point transmission, deep unfolding based precoding, etc. These topics are left for our future research.  
\appendices


\section{Proof of Proposition \ref{proposition_column_space} }\label{appendix_proposition_column_space}
We prove Proposition \ref{proposition_column_space} by performing a thorough analysis on the Karush-Kuhn-Tucker (KKT) conditions of problem \eqref{wsr_total_problem}. Before giving the formal proof, we first prove a basic fact stated in the following lemma.

\begin{lemma}\label{lemma_multiplier}\!\!\emph{:}
For any nontrivial stationary point $\mathbf{P}^{\star}$ of problem \eqref{wsr_total_problem}, the corresponding Lagrange multiplier $\lambda^{\star}$ associated with the SPC must be positive, i.e., $\lambda^{\star}>0$, and vice versa. 
\end{lemma}
\begin{proof}
Let us first show the necessity. Note that the linear independence constraint qualification (LICQ) holds for all the feasible solutions since there is only a single constraint and the gradient of the constraint is not $\mathbf{0}$ when the inequality constraint is active. Therefore, for any nontrivial stationary point $\mathbf{P}^{\star}$, there exists a Lagrange multiplier $\lambda^{\star}$, together with $\mathbf{P}^{\star}$, satisfying the the KKT condition of problem \eqref{wsr_total_problem} as follows.
\begin{subequations}
\begin{align}
&\alpha_{k} \nabla_{\mathbf{P}_{k}} R_{k}+\sum_{i \neq k}^{K} \alpha_{i} \nabla_{\mathbf{P}_{k}} R_{i}-\lambda^{\star} \mathbf{P}_{k}^{\star}=\mathbf{0}, \forall k,\label{reduced_kkt}
\\
&\left(\sum_{k=1}^{K} \operatorname{Tr}\left(\mathbf{P}^{\star}_{k} (\mathbf{P}^{\star}_{k})^{H}\right) - P_{\max}\right)\cdot \lambda^{\star} =0.\label{complementary_slackness}\\
&\sum_{k=1}^{K} \operatorname{Tr}\left(\mathbf{P}^{\star}_{k} (\mathbf{P}^{\star}_{k})^{H}\right) \leq P_{\max},  \label{wsr_total_constraints_redefine}\\
&\lambda^{\star} \geq 0,\label{dual_feasibility}
\end{align}
\end{subequations}
where \eqref{reduced_kkt} is the first-order optimality conditions with respect to the precoders. \eqref{complementary_slackness} is the complementary slackness condition; \eqref{wsr_total_constraints_redefine} and \eqref{dual_feasibility} are the primal and dual feasibility conditions, respectively.

Next, we prove $\lambda^{\star} > 0$ by contradiction. Assume the contrary that $\lambda^{\star} =0$.
For the convenience of gradient derivation, we rewrite the achievable rate of user $k$ in \eqref{single_user_rate} as
\begin{equation}\label{Rk_unfold}
\begin{aligned}
R_{k}
=&\log \operatorname{det}\left(\sum_{j=1}^{K} \mathbf{H}_{k} \mathbf{P}_{j} \mathbf{P}_{j}^{H} \mathbf{H}_{k}^{H} + \sigma_{k}^{2}\mathbf{I}\right)\\
&-\log \operatorname{det}\left(\sum_{j \neq k}^{K} \mathbf{H}_{k} \mathbf{P}_{j} \mathbf{P}_{j}^{H} \mathbf{H}_{k}^{H}+\sigma_{k}^{2}\mathbf{I}\right).
\end{aligned}
\end{equation}
Define $\nabla_{\mathbf{P}_{k} }\mathcal{L}=\frac{ \partial \mathcal{L} }{ \partial \mathbf{P}_{k}^{\ast} }$
as the complex gradient operator. Then based on the result in \cite{petersen2008matrix}, i.e., $\nabla\log\det \mathbf{X}=\tr\left(\mathbf{X}^{-1}\nabla\mathbf{X}\right)$, we calculate the gradient of $R_{k}$ with respect to $\mathbf{P}_k$, yielding
\begin{equation}\label{gradient_P_k_R_k}
\nabla_{\mathbf{P}_{k} }R_{k}=\mathbf{H}_{k}^{H}\mathbf{Z}_{kk}, 
\end{equation}
where
\begin{equation}\label{define_Zkk}
\mathbf{Z}_{kk}\triangleq \left(\sum_{j=1}^{K} \mathbf{H}_{k} \mathbf{P}_{j} \mathbf{P}_{j}^{H} \mathbf{H}_{k}^{H}+\sigma_{k}^{2}\mathbf{I}\right)^{-1} \mathbf{H}_{k} \mathbf{P}_{k}.
\end{equation}
Further, by taking the gradient of $R_{i}, i \neq k$, with respect to $\mathbf{P}_k$, we have
\begin{equation}\label{gradient_P_k_R_i}
\nabla_{\mathbf{P}_{k} }R_{i}=\mathbf{H}_{i}^{H}\mathbf{Z}_{ik}, 
\end{equation}
where
\begin{equation}\label{define_Zik}
\begin{aligned}
\mathbf{Z}_{ik}\triangleq & \left(\sum_{j=1}^{K} \mathbf{H}_{i} \mathbf{P}_{j} \mathbf{P}_{j}^{H} \mathbf{H}_{i}^{H}+\sigma_{i}^{2}\mathbf{I}\right)^{-1} \mathbf{H}_{i} \mathbf{P}_{k}\\
&-\left(\sum_{j\neq i}^{K} \mathbf{H}_{i} \mathbf{P}_{j} \mathbf{P}_{j}^{H} \mathbf{H}_{i}^{H}+\sigma_{i}^{2}\mathbf{I}\right)^{-1} \mathbf{H}_{i} \mathbf{P}_{k}.
\end{aligned}
\end{equation}

Left-multiplying \eqref{reduced_kkt} by $(\mathbf{P}_k^{\star})^H$ yields
\begin{equation}\label{eq:left42}
\alpha_{k} (\mathbf{P}_k^{\star})^H\nabla_{\mathbf{P}_{k}} R_{k}+\sum_{i \neq k}^{K} (\mathbf{P}_k^{\star})^H\alpha_{i} \nabla_{\mathbf{P}_{k}} R_{i}=\mathbf{0},~\forall k,
\end{equation}
where we have used the assumption $\lambda^{\star}=0$. Summing \eqref{eq:left42} over $k=1,2,\ldots,K$ and rearranging the terms, we further obtain \eqref{eq_ap_proof_summing} in the top of the next page.
 
 \begin{figure*}[ht] 
 	\centering
 	\begin{equation}\label{eq_ap_proof_summing}
 	\begin{aligned}
 	\sum_{k=1}^{K}\sum_{i=1}^{K}\alpha_k(\mathbf{P}_i^{\star})^H\mathbf{H}_{k}^{H} \left(\sum_{j=1}^{K} \mathbf{H}_{k} \mathbf{P}_{j}^{\star} (\mathbf{P}_j^{\star})^H \mathbf{H}_{k}^{H}+\sigma_{k}^{2}\mathbf{I}\right)^{-1}& \mathbf{H}_{k} \mathbf{P}_{i}^{\star}\\
=\sum_{k=1}^{K}&\sum_{i\neq k}^{K}\alpha_k(\mathbf{P}_{i}^{\star})^{H}\mathbf{H}_{k}^{H} \left(\sum_{j\neq k}^{K} \mathbf{H}_{k} \mathbf{P}_{j}^{\star} (\mathbf{P}_{j}^{\star})^{H} \mathbf{H}_{k}^{H}+\sigma_{k}^{2}\mathbf{I}\right)^{-1} \mathbf{H}_{k} \mathbf{P}_{i}^{\star}.
\end{aligned}
 	\end{equation}
 	\hrulefill
 \end{figure*}

Taking trace on both sides of \eqref{eq_ap_proof_summing} and moreover utilizing the identities $\tr(\mathbf{A}\mathbf{B})=\tr(\mathbf{B}\mathbf{A})$ and $\tr(\mathbf{A}(\mathbf{A}+\mathbf{I})^{-1})=\tr(\mathbf{I})-\tr(\mathbf{A}+\mathbf{I})^{-1}$, we have
\begin{equation}\label{eq_derevation_lemma2}
\begin{aligned}
&\sum_{k=1}^{K}\alpha_k \sigma_{k}^{2}\tr\left(\sum_{j=1}^{K} \mathbf{H}_{k} \mathbf{P}_{j}^{\star} (\mathbf{P}_{j}^{\star})^{H} \mathbf{H}_{k}^{H}+\sigma_{k}^{2}\mathbf{I}\right)^{-1} \\
=&\sum_{k=1}^{K}\alpha_k \sigma_{k}^{2}\tr\left(\sum_{j\neq k}^{K} \mathbf{H}_{k} \mathbf{P}_{j}^{\star} (\mathbf{P}_{j}^{\star})^{H} \mathbf{H}_{k}^{H}+\sigma_{k}^{2}\mathbf{I}\right)^{-1}.
\end{aligned}
\end{equation}
Note that the above equation \eqref{eq_derevation_lemma2} holds if and only if $\mathbf{H}_k\mathbf{P}_k^{\star}=\mathbf{0},\forall k$, which contradicts the fact that $\mathbf{P}^{\star}$ is a nontrivial stationary point. Therefore, we conclude that, for any nontrivial stationary point $\mathbf{P}^{\star}$, the corresponding $\lambda^{\star}$ must be positive.

Conversely, we will prove that if the Lagrange multiplier $\lambda^{\star}>0$, then the corresponding $\mathbf{P}^{\star}$ must be a nontrivial stationary point. We can prove this by contraction. Assume for the contrary that $\mathbf{P}^{\star}$ is a trivial stationary point, i.e., such that $\mathbf{H}_k\mathbf{P}_k^{\star}=\mathbf{0},\forall k$. It follows that $\mathbf{Z}_{kk}=0$ and $\mathbf{Z}_{ik}=0, \forall i\neq k$, with which, \eqref{reduced_kkt} can be reduced to $\lambda^{\star} \mathbf{P}_{k}^{\star}=\mathbf{0}, \forall k$. This, together with $\lambda^{\star}>0$, implies $\mathbf{P}^{\star}=\mathbf{0}$, which contradicts with  \eqref{complementary_slackness}. Thus the proof of Lemma  \ref{lemma_multiplier} is finished.
\end{proof}

Based on Lemma \ref{lemma_multiplier}, we prove Proposition \ref{proposition_column_space} as follows. By substituting $\nabla_{\mathbf{P}_{k} }R_{k}=\mathbf{H}_{k}^{H}\mathbf{Z}_{kk}$ and $\nabla_{\mathbf{P}_{k} }R_{i}=\mathbf{H}_{i}^{H}\mathbf{Z}_{ik}$ into \eqref{reduced_kkt} and noting $\lambda^{\star}>0$, we infer from \eqref{reduced_kkt} 
\begin{equation}
\mathbf{P}_{k}^{\star}=\frac{1}{\lambda^{\star}}\left(\alpha_{k} \mathbf{H}_{k}^{H} \mathbf{Z}_{k k}+\sum_{i \neq k}^{K} \alpha_{i} \mathbf{H}_{i}^{H} \mathbf{Z}_{i k}\right),~ \forall k,
\end{equation}
which implies that any nontrivial stationary point $\mathbf{P}_{k}^{\star}$ must lie in the range space of $\mathbf{H}^{H}$. Thus the proof of Proposition \ref{proposition_column_space} is completed.

\section{EZF Precoding}\label{appendix_ezf}
EZF precoding \cite{sun2010eigen} can effectively cancel multi-user interference and is generally more powerful than ZF. It can be obtained by performing ZF on an equivalent channel based on singular value decomposition (SVD), detailed below. 

To do EZF, we first perform thin SVD on $\mathbf{H}_{k}$ for all user $k$, i.e., $\mathbf{H}_{k}= \check{\mathbf{U}}_{k} \mathbf{\Sigma}_{k} \mathbf{V}_{k}^{H},\ \forall k$, where $\mathbf{\Sigma}_{k} \in \mathbb{C}^{N_k\times N_k}$ is a diagonal matrix with positive singular values sorted in a descending order, $ \check{\mathbf{U}}_{k} \in \mathbb{C}^{N_k\times N_k}$ is a unitary matrix consisting of left singular vectors, and 
$\mathbf{V}_{k} \in \mathbb{C}^{M\times N_k}$ consists of right singular vectors stored in each column. Then from $\mathbf{V}_{k}$ we choose $D_k$ singular vectors corresponding to the first $D_k$ largest singular values of $\mathbf{H}_{k}$, which is denoted by $\tilde{\mathbf{V}}_k \in \mathbb{C}^{M \times D_k }$. Define $\tilde{\mathbf{V}}\triangleq\left[\tilde{\mathbf{V}}_{1}, \tilde{\mathbf{V}}_{2}, \ldots, \tilde{\mathbf{V}}_{K}\right]\in \mathbb{C}^{M \times D }$, then $\tilde{\mathbf{V}}^H$ can be viewed as an equivalent channel. Finally, by performing ZF precoding on the equivalent channel $\tilde{\mathbf{V}}^H$, we obtain the EZF precoding: $\mathbf{P}_{\text{EZF}} =\tilde{\mathbf{V}}\left(\tilde{\mathbf{V}}^{H} \tilde{\mathbf{V}}\right)^{-1}$.

Since $\mathbf{V}_{k}$ and $\mathbf{H}_{k}^{H}$ have the same range space, thus $R(\tilde{\mathbf{V}}_k)  \subset R(\mathbf{H}_{k}^{H})$ must holds for all $k$. Therefore, we have $R(\tilde{\mathbf{V}}) \subset R(\mathbf{H}^{H})$, which means EZF precoder also follows the low-dimensional subspace property.

\section{Proof of Proposition \ref{proposition_unconstrained}}\label{appendix_proposition_unconstrained}
We first show the sufficiency. Let $\mathbf{P}^{\ddagger}=[\mathbf{P}_{1}^{\ddagger},\mathbf{P}_{2}^{\ddagger},\ldots,\mathbf{P}_{K}^{\ddagger}]$
be any stationary point of problem \eqref{unconstrained_WSR_problem} and define $\mathbf{P}_{k}^{\star}\triangleq \sqrt{\omega}\mathbf{P}_{k}^{\ddagger},\forall k$ with $\omega=\frac{P_{\max}}{\sum_{k=1}^{K}\operatorname{Tr}\left(\mathbf{P}_{k}^{\ddagger} (\mathbf{P}_{k}^{\ddagger})^{H}\right) }$. Then the sufficiency is to show that the point $\mathbf{P}^{\star}=[\mathbf{P}_{1}^{\star},\mathbf{P}_{2}^{\star},\ldots,\mathbf{P}_{K}^{\star}]$ is a nontrivial stationary point of problem \eqref{wsr_total_problem}. It is proven below by comparing the KKT conditions of the two problems.



First, since $\mathbf{P}^{\ddagger}$ is a stationary point of the unconstrained problem \eqref{unconstrained_WSR_problem}, we have
\begin{equation}\label{phi_1_P_grad}
\begin{aligned}
      \alpha_{k}\nabla_{\mathbf{P}_{k} }\tilde{R}_k(\mathbf{P}^{\ddagger}) + \sum_{i\neq k}\alpha_{i}\nabla_{\mathbf{P}_{k} }\tilde{R}_i(\mathbf{P}^{\ddagger})=\mathbf{0}, \forall k,
\end{aligned}
\end{equation}
where
\begin{equation}\label{R_k_tilde}
\begin{aligned}
      &\tilde{R}_k(\mathbf{P})  \triangleq   \log \operatorname{det}\Bigg(\mathbf{I}+ \mathbf{H}_{k} \mathbf{P}_{k} \mathbf{P}_{k}^{H} \mathbf{H}_{k}^{H}\\
       &\ \ \Big( \sum_{j\neq k} \mathbf{H}_{k} \mathbf{P}_{j} \mathbf{P}_{j}^{H} \mathbf{H}_{k}^{H}+\frac{\sigma_{k}^{2}}{P_{\max}}\sum_{i=1}^{K} \operatorname{Tr}( \mathbf{P}_{i} \mathbf{P}_{i}^{H}) \mathbf{I}\Big)^{-1}\Bigg), \forall k.
\end{aligned}
\end{equation}

Next, let us figure out the gradients above. Note that $\tilde{R}_k(\mathbf{P})  =  \log \operatorname{det}\left(\mathbf{D}_k\right){-}\log \operatorname{det}\left(\mathbf{F}_k\right)$ with\footnote{$\mathbf{D}_k$ is a matrix function of $\mathbf{P}$ but we drop the argument for the convenience of notation. Similarly for $\mathbf{F}_k$.}
\begin{equation}
\mathbf{D}_k \triangleq \Big( \sum_{j= 1}^{K} \mathbf{H}_{k} \mathbf{P}_{j} \mathbf{P}_{j}^{H} \mathbf{H}_{k}^{H}+\frac{\sigma_{k}^{2}}{P_{\max}}\sum_{i=1}^{K} \operatorname{Tr}(\mathbf{P}_{i} \mathbf{P}_{i}^{H}) \mathbf{I}\Big),
\end{equation}
\begin{equation}
\mathbf{F}_k \triangleq \Big( \sum_{j\neq k} \mathbf{H}_{k} \mathbf{P}_{j} \mathbf{P}_{j}^{H} \mathbf{H}_{k}^{H}+\frac{\sigma_{k}^{2}}{P_{\max}}\sum_{i=1}^{K} \operatorname{Tr}( \mathbf{P}_{i} \mathbf{P}_{i}^{H}) \mathbf{I}\Big).
\end{equation}
Taking the gradient of $\tilde{R}_k$ with respect to $\mathbf{P}_k$ yields
\begin{equation}\label{R_k_tilde_gradient1}
\begin{aligned}
\nabla_{\mathbf{P}_{k} }\tilde{R}_k(\mathbf{P}) =&\mathbf{H}_{k}^{H}\mathbf{D}_{k}^{-1}\mathbf{H}_{k}\mathbf{P}_{k} + \frac{\sigma_{k}^{2}}{P_{\max}}  \operatorname{Tr}(\mathbf{D}_{k}^{-1})\mathbf{P}_{k}\\
&-\frac{\sigma_{k}^{2}}{P_{\max}}  \operatorname{Tr}(\mathbf{F}_{k}^{-1})\mathbf{P}_{k}.
\end{aligned}
\end{equation}
Further, by taking the gradient of $\tilde{R}_{i}, i \neq k$ with respect to $\mathbf{P}_k$, we have
\begin{equation}\label{R_k_tilde_gradient2}
\begin{aligned}
\nabla_{\mathbf{P}_{k} }\tilde{R}_i(\mathbf{P}) =&\left(\mathbf{H}_{i}^{H}\mathbf{D}_{i}^{-1}\mathbf{H}_{i}\mathbf{P}_{k} + \frac{\sigma_{i}^{2}}{P_{\max}}  \operatorname{Tr}(\mathbf{D}_{i}^{-1})\mathbf{P}_{k}\right)\\
 &-\left(\mathbf{H}_{i}^{H}\mathbf{F}_{i}^{-1}\mathbf{H}_{i}\mathbf{P}_{k} + \frac{\sigma_{i}^{2}}{P_{\max}}  \operatorname{Tr}(\mathbf{F}_{i}^{-1})\mathbf{P}_{k}\right).
\end{aligned}
\end{equation}

Note that for any $t>0$ we have $ \nabla_{\mathbf{P}_{k} }\tilde{R}_i(t\mathbf{P})=\frac{1}{t}\nabla_{\mathbf{P}_{k} }\tilde{R}_i(\mathbf{P}), i \neq k$ and $\nabla_{\mathbf{P}_{k} }\tilde{R}_k(t\mathbf{P})=\frac{1}{t}\nabla_{\mathbf{P}_{k} }\tilde{R}_k(\mathbf{P})$. By using this fact, \eqref{phi_1_P_grad} is equivalent to
\begin{equation}\label{phi_1_P_grad_recast}
      \alpha_{k}\nabla_{\mathbf{P}_{k} }\tilde{R}_k(\sqrt{\omega}\mathbf{P}^{\ddagger}) {+} \sum_{i\neq k}\alpha_{i}\nabla_{\mathbf{P}_{k} }\tilde{R}_i(\sqrt{\omega}\mathbf{P}^{\ddagger}){=}\mathbf{0}, \forall k.
\end{equation}
Further, by noting  $\sum_{k=1}^{K}\operatorname{Tr}\left(\sqrt{\omega}\mathbf{P}_{k}^{\ddagger} (\sqrt{\omega}\mathbf{P}_{k}^{\ddagger})^{H}\right) =P_{\max}$ and $\mathbf{P}_{k}^{\star}=\sqrt{\omega}\mathbf{P}_{k}^{\ddagger},\forall k$, 
we can recast \eqref{phi_1_P_grad_recast} as  
\begin{equation}\label{X_k_expression}
\alpha_{k} \mathbf{H}_{k}^{H}\mathbf{Z}_{kk} +\sum_{i\neq k}^{K}\alpha_{i} \mathbf{H}_{i}^{H}\mathbf{Z}_{ik} - \gamma \mathbf{P}_k^{\star} = \mathbf{0}, \forall k,
\end{equation}
where $\mathbf{Z}_{kk}$ and $\mathbf{Z}_{ik}$ are respectively given in \eqref{define_Zkk} and \eqref{define_Zik} with thereof $\mathbf{P}$ replaced by $\mathbf{P}^{\star}$, and $\gamma = \sum_{i=1}^{K}\frac{\alpha_{i}\sigma_{i}^{2}}{P_{\max}}\left(\operatorname{Tr}(\mathbf{F}_{i}^{-1})-\operatorname{Tr}(\mathbf{D}_{i}^{-1})\right)>0$. Note that \eqref{X_k_expression} implies that $\mathbf{P}^{\star}$ satisfies \eqref{reduced_kkt} with $\gamma>0$ being the Lagrange multiplier. Moreover, complementary slackness condition \eqref{complementary_slackness} follows because $\mathbf{P}^{\star}$ meets the power constraints with equality. Therefore, $\mathbf{P}^{\star}$ satisfies the KKT conditions of the original WSR maximization problem \eqref{wsr_total_problem} and moreover it is a nontrivial stationary point of \eqref{wsr_total_problem}. This completes the proof of sufficiency. The necessity can be proven by reversing the steps of the sufficiency proof.

\section{Proof of Theorem \ref{th_convergence_result} }\label{appendix_theorem2}
Since the objective function of the unconstrained problem \eqref{transform_problem_X} is continuously differentiable and convex in each block variable. Following the classic optimization theory of BCD method \cite{bertsekas1999nonlinear}, the proposed R-WMMSE algorithm is guaranteed to converge to a stationary point of \eqref{transform_problem_X}. Let $\left(\mathbf{U}^{\star},\mathbf{W}^{\star},\mathbf{X}^{\star}\right)$ be any stationary point of \eqref{transform_problem_X}. Then it remains to show that $\mathbf{P}^{\star}=\sqrt{\beta}\mathbf{H}^{H}\mathbf{X}^{\star}$, where $\beta=\frac{P_{\max}}{\sum_{k=1}^{K}\operatorname{Tr}\left(\bar{\mathbf{H}}\mathbf{X}_{k}^{\star} (\mathbf{X}_{k}^{\star})^{H}\right)}$, is a nontrivial stationary point of \eqref{wsr_total_problem}.


Let $\phi_{1}\left(\mathbf{X}\right)$ and $\phi_{2}\left(\mathbf{U},\mathbf{W},\mathbf{X}\right)$ denote the objective functions of problems \eqref{unconstrained_WSR_problem_X} and \eqref{transform_problem_X}, respectively. Since $\left(\mathbf{U}^{\star},\mathbf{W}^{\star},\mathbf{X}^{\star}\right)$ is a stationary point of  \eqref{transform_problem_X}, we have
\begin{subequations}
\begin{align}
\label{Uk_stationary_condition}
&\nabla_{\mathbf{U}_{k}} \phi_{2}\left( \mathbf{U}^{\star},\mathbf{W}^{\star}, \mathbf{X}^{\star}\right) = \mathbf{0}, ~ \forall k,\\
\label{Wk_stationary_condition}
&\nabla_{\mathbf{W}_{k}} \phi_{2}\left( \mathbf{U}^{\star},\mathbf{W}^{\star}, \mathbf{X}^{\star}\right) = \mathbf{0},~  \forall k,\\
\label{Xk_stationary_condition}
&\nabla_{\mathbf{X}_{k}} \phi_{2}\left( \mathbf{U}^{\star},\mathbf{W}^{\star}, \mathbf{X}^{\star}\right) = \mathbf{0}, ~ \forall k.
\end{align}
\end{subequations}
\eqref{Uk_stationary_condition} and \eqref{Wk_stationary_condition} immediately implies that ${\mathbf{U}_{k}^{\star}}$ and ${\mathbf{W}_{k}^{\star}}$ must satisfy \eqref{U_update} and \eqref{W_update}, respectively. Further combining \eqref{Xk_stationary_condition} with Lemma \ref{thm_transform} (which establishes the relation between $\phi_{1}$ and $\phi_{2}$) yields
\begin{equation}\label{phi_3_X_grad}
\nabla_{\mathbf{X}_{k}} \phi_{1}\left( \mathbf{X}^{\star}\right) = \nabla_{\mathbf{X}_{k}} \phi_{2}\left( \mathbf{U}^{\star},\mathbf{W}^{\star}, \mathbf{X}^{\star}\right)=\mathbf{0},~  \forall k.
\end{equation}
Furthermore, by noting $\mathbf{P}^{\star}=\sqrt{\beta}\mathbf{H}^{H}\mathbf{X}^{\star}$ with $\beta=\frac{P_{\max}}{\sum_{k=1}^{K}\operatorname{Tr}\left(\bar{\mathbf{H}}\mathbf{X}_{k}^{\star} (\mathbf{X}_{k}^{\star})^{H}\right)}$ and  using similar idea for the proof of Proposition \ref{proposition_reduced_unconstrained}, we can recast $\eqref{phi_3_X_grad}$ as \eqref{reduced_kkt}. This completes the proof.

\bibliographystyle{IEEEtran}
\bibliography{to_bibitem}


 




\vfill

\end{document}